\documentclass[lettersize,journal]{IEEEtran}
\normalsize
\usepackage{cite}
\ifCLASSINFOpdf
  \usepackage[pdftex]{graphicx}
\else
\fi
\usepackage[cmex10]{amsmath}
\usepackage{amsfonts}
\allowdisplaybreaks[4]
\usepackage{algorithm}
\usepackage{algorithmicx}
\usepackage[noend]{algpseudocode}
\usepackage{array}
\usepackage{mdwmath}
\usepackage{mdwtab}
\usepackage{amsthm}
\usepackage{eqparbox}
\usepackage[tight,footnotesize]{subfigure}
\usepackage{amssymb}
\usepackage[tight,footnotesize]{subfigure}
\newtheorem{lemma}{Lemma}
\newtheorem{corollary}{Corollary}

\begin{document}
\title{On the Secrecy Rate of In-Band Full-duplex Two-way Wiretap Channel}
\author{Navneet Garg, \IEEEmembership{Senior Member, IEEE}, 
	Haifeng Luo, \IEEEmembership{Member, IEEE},
    and Tharmalingam Ratnarajah, \IEEEmembership{Senior Member, IEEE}
\thanks{N. Grag, H. Luo, and T. Ratnarajah are with Institute for Digital Communications, School of Engineering, The University of Edinburgh, UK, e-mails: navneet.garg4@gmail.com, h.luo@ieee.org, t.ratnarajah@ieee.org \\ 
This paper was presented in part at IEEE 23rd International Workshop on Signal Processing Advances in Wireless Communication (SPAWC), Oulu, Finland, 2022 \cite{9833919}.}
}

\markboth{IEEE Transactions on Green Communications and Networking}%
{Draft}

\maketitle

\begin{abstract}
In this paper, we consider a two-way wiretap Multi-Input Multi-Output
Multi-antenna Eve (MIMOME) channel, where both nodes (Alice and Bob)
transmit and receive in an in-band full-duplex (IBFD) manner. For this system with keyless
security, we provide a novel artificial noise (AN) based signal design, where the AN is injected in both signal and null spaces. We
present an ergodic secrecy rate approximation to derive the power allocation algorithm. We consider scenarios where AN is known and unknown to legitimate users and include imperfect channel information effects. To maximize secrecy rates subject to the transmit
power constraint, a two-step power allocation solution is proposed,
where the first step is known at Eve, and the second step helps to improve
the secrecy further. We also consider scenarios where partial information is known by Eve and the effects of non-ideal self-interference cancellation. The usefulness
and limitations of the resulting power allocation solution are analyzed and verified
via simulations. Results show that secrecy rates are less
when AN is unknown to receivers or Eve has more information about legitimate users. Since the ergodic approximation only
considers Eve's distance, the resulting power allocation
provides secrecy rates close to the actual ones. 
\end{abstract}


\IEEEpeerreviewmaketitle
\section{Introduction}
\label{sec:intro}
\IEEEPARstart{I}{n}formation security in wireless systems is essential in today's world
as quick information exchange
is influencing lives. Among several key issues in the security of wireless
networks, sharing secret keys is most common in cryptography. However,
establishing any secret key between two legitimate users is not trivial.
The lifetime of a secret key is shortened each time the key is used.
For many applications, such as big data streaming, periodic renewal
is required for secret keys. Therefore, generating secret keys and
transmitting secure information using a wireless medium are the two
parts of physical layer security \cite{8883079,sarkar2012enhancing}. Based on that, there
are two types of secrecy methods in the literature, namely, key-based
and keyless. This paper investigates the keyless transmission of
secret information between two users using artificial noise (AN)
\cite{7505974,7247765}. Specifically, two legitimate users (Alice and Bob) are considered to exchange secret information while being
eavesdropped by Eve. When users have multiple antennas,
this setup is called Multi-Input Multi-Output Multi-antenna
Eve (MIMOME) channel. 

In \cite{YANG2019100730,8709756}, with in-band full-duplex (IBFD) radios at the legitimate receiver, the secrecy rate is maximized for a one-way wiretap channel using block coordinate descent method with respect to the transmit
covariance matrix. In this one-way wiretap channel, the IBFD node transmits
jamming signals to reduce Eve's rates. In \cite{7925740,8883079} in addition to jamming signals, artificial noise is also transmitted
from Alice, and the tight lower bound on the achievable ergodic
secrecy rate is derived. In \cite{8355793}, Bob transmits secure pilot symbols
to estimate the self-interference (SI) channel and make Eve unable to acquire
Bob-to-Eve channel estimates (jamming channel). A similar scenario
with multiple randomly located Eve is analyzed using Poisson point
process in \cite{9309259}. Also, with multiple Bobs and Eves, the power
allocation problem is tacked in \cite{7792199}. In \cite{8265212}
to improve achievable secrecy degrees of freedom (sDoFs), the antenna
assignment problem at Bob is investigated with SI present,
and a feasibility condition is derived. 

In a two-way wiretap channel, where Alice and Bob operate
in an IBFD manner, the loss of secrecy rates for feedback is investigated
for a single antenna scenario in \cite{5074582}. Transmit covariance
matrix is optimized for robust precoders with perfect CSI in \cite{6808487} and imperfect CSI with bounded uncertainty in \cite{7842236,8288039,7015632,7414075}.
A similar design is adopted for simultaneous wireless information
and power transfer (SWIPT) systems toward weighted sum rate maximization.
Secrecy region investigation in \cite{6600802} reveals that randomized
scheduling can improve secrecy rates, assuming only Eve's location
information. One-sided secrecy of two-way channels is analyzed for
adder and binary XOR channels in \cite{7840500}, and it shows that channel prefixing
can improve secrecy rates. In \cite{7945480}, a close form approximation
for ergodic secrecy rate and sDoFs is derived in a single antenna scenario. In \cite{7942069}, the secrecy outage probability is studied and the maximum tolerable number of eavesdroppers is given. There are works considering a two-way IBFD relay, such as \cite{6730702}; however, they are different from the two-way wiretap scenario, where two nodes operate in IBFD manner. In our previous work, we have derived expressions for two-way power
allocation \cite{9833919}, which can be considered as a special case
of this work.

In this paper, we consider a general two-way MIMOME wiretap channel,
that is, a one-way wiretap channel can be derived assuming zero desired
signal transmit power at Bob. Other features that are considered separately
in above works are also included together in this work, such as imperfect
CSI, imperfect self-interference cancellation, multiple antennas at
each node, only Eve's location is known at Alice and Bob, secrecy region
details, ergodic secrecy rate  approximation and its maximization,
and fast algorithmic solution. With all these considerations, we derive a general power allocation algorithm to maximize the secrecy rates of the two-way MIMOME wiretap channel. Moreover, in literature, artificial
noise is only injected into the null space, while we give formulation with artificial
noise in both signal and null spaces, including precise stream-wise power controls.
Details of the contributions are given as follows.

\subsection{Contributions}
\begin{itemize}
      \item AN-based signal design: In contrast to works in literature, artificial noise is added in both
      the signal and null spaces. Moreover, we
      consider stream-wise power control in both signal and AN directions.
      This not only allows better power allocation to increase Alice-to-Bob
      rates, but also improves the secrecy rates since this fine stream-wise
      allocation is unknown to Eve. 
      \item Secrecy rate approximation and secrecy regions with partial precoder
      knowledge at Eve: Since Alice and Bob are only aware of Eve's distance information,
      the ergodic rates at Eve are approximated, and the corresponding limitations
      are analyzed in simulations. It can be noted that the imperfect channel
      estimation is considered by Alice and Bob, and thus, Eve also does
      not have perfect legitimate channel knowledge. This information is
      quantified, and its effects are also analyzed with the help of chordal
      distance decomposition. Secrecy regions for both Alice and Bob are
      plotted for different cases to gain further insights. 
      \item Power allocation with and without AN knowledge: The power allocation problem is formulated to maximize sum ergodic rates,
      and a two-step solution is provided. The first coarse allocation step
      is based on ergodic rate approximation, which includes large-scale fading
      and system parameters. For considering a practical scenario, we assume
      the outcome of this step is also known by Eve. However, the second
      fine allocation step is based on small-scale channel fading matrix
      between Alice and Bob, making this step unknown to Eve. The power
      allocation solution is obtained separately for both cases when AN
      is known at Alice/Bob and when not. Successive convex approximation
      provides a gradient descent method for the known AN case, whereas Newton's
      second-order method for the unknown AN case. These methods are fast
      and converge in a few iterations ($<20$ for present simulations). 
      \item Simulations with imperfect CSI and imperfect SIC: Simulation verifies the usefulness of the power allocation to improve secrecy
      rates. Simulations are shown for sum secrecy rates with respect to
      transmit power, partial information at Eve, power allocation factors,
      Eve's location, residual self-interference strength, partial information
      content at Eve, and Eve rates for regulating AN in precoder and its
      null directions. Partial information at Eve certainly improves the
      secrecy rates, whereas unknown AN reduces the same. 
\end{itemize}

\subsection{Orgnization}
The rest of the paper is organized as follows: Section \ref{sec:preliminaries} describes the details of the two-way
wiretap system. The proposed power allocation and its analysis are
described in Section \ref{sec:Two-step-secrecy-power}. Simulation
results are demonstrated in Section \ref{sec:Simulation-Results},
followed by the conclusion in Section \ref{sec:Conclusion}. 

\subsection{Notations}
Scalars, vectors, and matrices are represented by lower case ($a$),
lower case bold face ($\mathbf{a}$) and upper case bold face ($\mathbf{A}$)
letters, respectively. Conjugate, transpose, Hermitian transpose,
element-wise Hadamard product, and Kronecker product of matrices are
denoted by $(\cdot)^{*}$, $(\cdot)^{T}$, $(\cdot)^{\dagger}$, $\odot$
and $\otimes$, respectively. $\mathcal{CN}(\mathbf{\mu},\mathbf{R})$
represents a circularly symmetric complex Gaussian random vector with
mean $\mathbf{\mu}$ and covariance matrix $\mathbf{R}$. The notations
$\|\cdot\|_{2}$ and $\|\cdot\|_{F}$ denote the $l_{2}$ norm and
Frobenious norm, respectively. The notation $\text{vec}(\mathbf{X})$
denotes the vector obtained by stacking the columns of matrix $\mathbf{X}$.
$\mathcal{D}(a_{1},a_{2})$ and $\mathcal{BD}(\mathbf{A}_{1},\mathbf{A}_{2})$
denote a block diagonal matrix with respectively scalars $a_{1},a_{2}$
and matrices $\mathbf{A}_{1},\mathbf{A}_{2}$ as its diagonal components.
The notation $\mathbf{A}(1:d)$ signifies the matrix constructed
from the first $d$-columns of $\mathbf{A}$.

\section{System Model}
\label{sec:preliminaries}
We consider a two-way wiretap channel, as shown in Figure \ref{fig:Two-way-wiretap-channel-1},
where Alice (with $N_{A}$ antennas) intends to send secret information
over a wireless channel to Bob (with $N_{B}$ antennas) in the presence
of a passive eavesdropper (called Eve having $N_{E}$ antennas), which is also known as MIMOME system. 
\begin{figure}
	\centering \includegraphics[width=1.0\columnwidth]{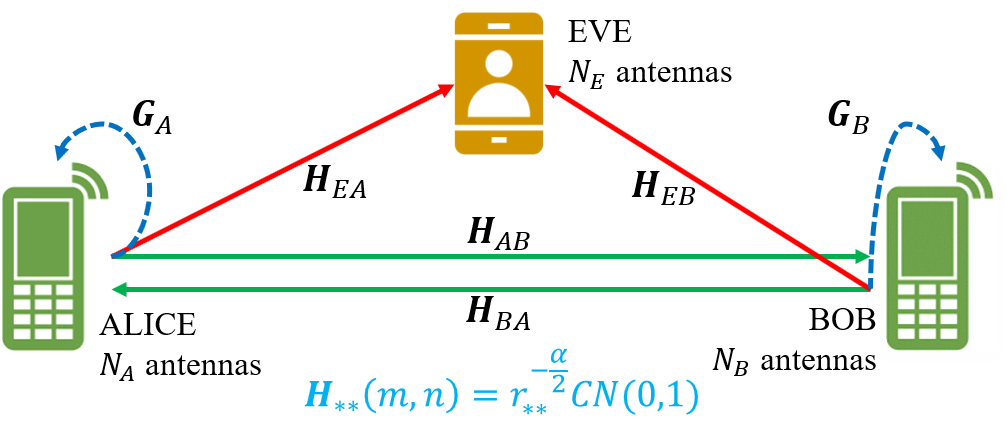}
	
	\caption{Two-way wiretap channel model with Alice and Bob operating in IBFD manner.\label{fig:Two-way-wiretap-channel-1}}
\end{figure}

\subsection{Channel model}
\subsubsection{Legitimate and eavesdropping channels}
We consider half-wavelength antenna arrays so that the entries of the MIMO channel matrix are independent of each other \cite{8794743}. Entries of true channel matrix following Rayleigh fading from the
node $i$ to the node $j$, which are given as 
\begin{equation}
	\mathbf{H}_{ji}(m,n)\sim\mathcal{CN}\left(0,\beta_{ji}\right),~i\neq j,
\end{equation}
where $i\in\left\{ A,B\right\} ,j\in\left\{ A,B,E\right\} ,$ $m=1,\ldots,N_{j}$;
$n=1,\ldots,N_{i}$; $\beta_{ji}=r_{ji}^{-\alpha}$; $r_{ji}$ is the distance between the two nodes; $\alpha$ is the path loss exponent. Channel reciprocity
is assumed to be present between true channels, that is, $\mathbf{H}_{AB}=\mathbf{H}_{AB}^{H}$. The true channel reciprocity assumption is reasonable due to the electromagnetic foundations \cite{1580281} and is used in most related studies (see \cite{luo2023channel} and references therein). However, legitimate users cannot fully utilize the channel reciprocity due to practical imperfections. Accurate channel estimation is challenging due to limited training resources and non-symmetric hardware impairments at the transmitter and receiver, resulting in channel uncertainty in the obtained channel state information (CSI). Let $\hat{\mathbf{H}}_{ji}$ denote the estimate associated with the actual wireless channel $\mathbf{H}_{ji}$, they are related as \cite{9374975}
\begin{equation}
	\mathbf{H}_{ji}=\hat{\mathbf{H}}_{ji}+\Delta_{ji},
\end{equation}
where $\Delta_{ji}$ denotes the random channel uncertainty matrix. In this paper, we adapt the stochastic error model, which describes the channel uncertainty by the complex Gaussian model, i.e., entries of the channel uncertainty matrix are given as $\mathbf{\Delta}_{ji}(m,n) \sim \mathcal{CN} \left( 0, \sigma_{\Delta,ji}^{2} \right)~\forall~m,n$. For the worst case scenario, it is assumed that Eve has the perfect knowledge of its matrices, $\mathbf{H}_{EA}$ and $\mathbf{H}_{EB}$.
On the other hand, Alice and Bob are only aware of Eve's location,
i.e., $\beta_{EA}$ and $\beta_{EB}$.

\subsubsection{Residual self-interference channels}
Since both Alice and Bob are operating in IBFD manner, self-interference cancellation (SIC) has to be applied before any effective IBFD operations. The implementation of the SIC is out of the scope of this paper, while the readers are referred to our previous work in \cite{luo2023channel} for the implementation and detailed effects of SIC with antenna arrays. With the same number of receive antennas as the transmit ones, the residual
self-interference channels (RSICs) after three-stage (passive, analog
and digital) cancellations are denoted by $\mathbf{G}_{A}\in\mathbb{C}^{N_{A}\times N_{A}}$ and $\mathbf{G}_{B}\in\mathbb{C}^{N_{B}\times N_{B}}$, respectively, where entries of $\mathbf{G}_{A}$ and $\mathbf{G}_{B}$ are assumed independent and identically distributed (i.i.d.) zero mean Gaussian distributed with the variance $\eta$. After three stage cancellation, the strength of RSICs is close to the noise variance, i.e., $\eta \approx \sigma^{2}$. 

\subsection{Precoding and transmitted signals}
Alice and Bob transmit the $b$ secret data streams precoded and mixed
with artificial noise as 
\begin{align}
	\mathbf{x}_{i} & =\mathbf{V}_{i}\left(\mathbf{P}_{s,i}^{\frac{1}{2}}\mathbf{s}_{i}+\mathbf{P}_{w,i}^{\frac{1}{2}}\mathbf{w}_{i}\right)+\underline{\mathbf{V}}_{i}\mathbf{P}_{\underline{w},i}^{\frac{1}{2}}\underline{\mathbf{w}}_{i},\quad i\in\left\{ A,B\right\} ,
\end{align}
where $\mathbf{V}_{i}$ is a $N_{i}\times b$ orthonormal precoder
matrix such that $\mathbf{V}_{i}^{H}\mathbf{V}_{i}=\mathbf{I}_{b}$ \cite{cirik2016beamforming};
$\underline{\mathbf{V}}_{i}=\text{null}\left(\mathbf{V}_{i}\right)\in\mathbb{C}^{N_{i}\times n_{i}}$
with $n_{i}=N_{i}-b$ such that $\underline{\mathbf{V}}_{i}^{H}\underline{\mathbf{V}}_{i}=\mathbf{I}_{n_{i}}$
and $\mathbf{V}_{i}^{H}\underline{\mathbf{V}}_{i}=\mathbf{0}$; the
data symbol vector $\mathbf{s}_{i}\in\mathbb{C}^{b\times1}$ denotes zero
mean symbols with the covariance matrix $\mathbb{E}\left\{ \mathbf{s}_{i}\mathbf{s}_{i}^{H}\right\} =\mathbf{I}_{b}$;
vectors $\mathbf{w}_{i}$ and $\underline{\mathbf{w}}_{i}$ represent
zero mean artificial noise with covariance matrices, $\mathbb{E}\left\{ \mathbf{w}_{i}\mathbf{w}_{i}^{H}\right\} =\mathbf{I}_{b}$,
$\mathbb{E}\left\{ \underline{\mathbf{w}}_{i}\underline{\mathbf{w}}_{i}^{H}\right\} =\mathbf{I}_{n_{i}}$;
the matrices $\mathbf{P}_{s,i}$, $\mathbf{P}_{w,i}$, and $\mathbf{P}_{\underline{w},i}$
are power allocation diagonal matrices with non-negative real entries
at their diagonals. It is worth noting that injecting artificial noise into the signal space ($\mathbf{w}_i$) seems to be a counter common sense behavior, as it can degrade the quality of legitimate communication. However, it may at the same time cause more serious effects to the eavesdropper than the noise inserted into the null space ($\underline{\mathbf{w}}_i$). Therefore, the introduction of $\underline{\mathbf{w}}_i$ provides more possibilities for system security design. When the introduction of $\underline{\mathbf{w}}_i$ is counterproductive, our power allocation policy adaptively reduces its power budget. When $\mathbf{P}_{\underline{w},i} = \mathbf{0}$, this noise does not exist. At the $i^{th}$ transmitter ($i=A,B$), the transmit
covariance matrix can be given as 
\begin{equation}
	\mathbf{T}_{i}=\mathbb{E}\left\{ \mathbf{x}_{i}\mathbf{x}_{i}^{H}\right\} =\mathbf{V}_{i}\left(\mathbf{P}_{s,i}+\mathbf{P}_{w,i}\right)\mathbf{V}_{i}^{H}+\underline{\mathbf{V}}_{i}\mathbf{P}_{\underline{w},i}\underline{\mathbf{V}}_{i}^{H},
\end{equation}
where for the given power constraint $tr\left ( \mathbb{E}\left \{ \mathbf{x}_{i}\mathbf{x}_{i}^{H} \right \} \right ) \leq P_{i}$,
we have the constraints
\begin{equation}
	tr\left(\mathbf{P}_{s,i}+\mathbf{P}_{w,i}\right)+tr\left(\mathbf{P}_{\underline{w},i}\right)\leq P_{i},\forall i=A,B.\label{eq:txcon1}
\end{equation}

\subsection{Eve's knowledge about the precoder}
Note that precoders $\mathbf{V}_{A}$ and $\mathbf{V}_{B}$ are chosen
from the right singular vectors of channel estimates $\hat{\mathbf{H}}_{BA}$
and $\hat{\mathbf{H}}_{AB}$, respectively. Since the channel reciprocity
holds, both Alice and Bob are aware of both $\mathbf{V}_{A}$ and
$\mathbf{V}_{B}$. However, Eve does not have the channel knowledge
of either $\hat{\mathbf{H}}_{BA}$ or $\hat{\mathbf{H}}_{AB}$, we
assume Eve acquires these channel or precoder values with some accuracy.
Let $\hat{\mathbf{V}}_{AE}$ and $\hat{\mathbf{V}}_{BE}$ represent
the precoders at Eve. The closeness of these precoders to actual precoders can be obtained using the chordal distance decomposition as
\begin{equation}
	\mathbf{V}_{i}=\hat{\mathbf{V}}_{iE}\mathbf{Q}_{i}\mathbf{\Lambda}_{i}+\underline{\hat{\mathbf{V}}}_{iE}\underline{\mathbf{Q}}_{i}\underline{\mathbf{\Lambda}}_{i},~ i\in\left\{ A,B\right\}  ,
\end{equation}
where $d_{i}=\texttt{d}_c^{2}\left(\hat{\mathbf{V}}_{iE},\mathbf{V}_{i}\right)=tr(\mathbf{\Lambda}_{i}\mathbf{\Lambda}_{i}^{H}) = b - tr(\underline{\mathbf{\Lambda}}_{i}\underline{\mathbf{\Lambda}}_{i}^{H})$ is the chordal distance; $\mathbf{\Lambda}_{i}$ and $\underline{\mathbf{\Lambda}}_{i}$ are upper triangular matrix; $\underline{\hat{\mathbf{V}}}_{iE}\in\text{null}\left(\hat{\mathbf{V}}_{iE}\right)$; and $\mathbf{Q}_{i}\in\mathbb{C}^{d\times d}$, $\underline{\mathbf{Q}}_{i}\in\mathbb{C}^{n_{i}\times d}$ are orthonormal matrices. For perfect precoder information at Eve, the value $d_{i}=0$ can be substituted. The closeness of these precoders to the actual ones can be measured with the difference $\mathbf{\Delta}_{\mathbf{V},iE}=\mathbf{V}_{i}-\hat{\mathbf{V}}_{i,E}$
and the corresponding norm as 
\begin{align*}
	& \kappa_{i,E}=\left\Vert \mathbf{\Delta}_{\mathbf{V},iE} \right\Vert _{F}^{2}=tr\left[\left(\mathbf{V}_{i}-\hat{\mathbf{V}}_{i,E}\right)^{H}\left(\mathbf{V}_{i}-\hat{\mathbf{V}}_{i,E}\right)\right]\\
	& =tr\left( \mathbf{V}_{i}^{H}\mathbf{V}_{i} \right)+tr\left( \hat{\mathbf{V}}_{i,E}^{H}\hat{\mathbf{V}}_{i,E}\right)-2tr \left ( \Re \left \{ \hat{\mathbf{V}}_{i,E}^{H}\mathbf{V}_{i} \right \} \right ) \\
	& =2b-2tr \left ( \Re \left \{ \hat{\mathbf{V}}_{i,E}^{H}\mathbf{V}_{i} \right \} \right )\\
	& =2b\left(1-\frac{1}{b}tr \left ( \Re \left \{ \hat{\mathbf{V}}_{i,E}^{H}\mathbf{V}_{i} \right \} \right ) \right),
\end{align*}
where $\kappa_{i,E}=0$ implies the precoder $\mathbf{V}_{i}$ is
perfectly known at Eve, whereas the value $\kappa_{i,E}=2d$ shows
Eve selects completely orthogonal precoder of the matrix $\mathbf{V}_{i}$.
In simulations, the imperfect precoder at Eve is generated as follows.

Given $\kappa_{i,E}$ value, we compute its relation to the chordal
distance. Using chordal distance decomposition \cite{9548957,9447191},
$\hat{\mathbf{V}}_{i,E}$ and $\mathbf{V}_{i}$ have the relation
as 
\begin{equation}
	\hat{\mathbf{V}}_{i,E}=\mathbf{V}_{i}\sqrt{1-\frac{d_{i}}{b}}+\mathbf{V}_{i}^{\text{null}}\sqrt{\frac{d_{i}}{b}},\label{eq:CDdecomp}
\end{equation}
where $\mathbf{V}_{i}^{\text{null}}$ spans of the null-space of $\mathbf{V}_{i}$
such that $\mathbf{V}_{i}^{H}\mathbf{V}_{i}^{\text{null}}=\mathbf{0}$.
On multiplying with $\mathbf{V}_{i}^{H}$ gives 
\begin{gather}
	\mathbf{V}_{i}^{H}\hat{\mathbf{V}}_{i,E}=\mathbf{V}_{i}^{H}\mathbf{V}_{i}\sqrt{1-\frac{d_{i}}{b}}\\
	\kappa_{i,E}=2d-2d\sqrt{1-\frac{d_{i}}{b}},
\end{gather}
which provides 
\begin{equation}
	d_{i}=b\left[1-\left(1-\frac{\kappa_{i,E}}{2b}\right)^{2}\right].
\end{equation}
Therefore, given $\kappa_{i,E}$, compute $d_{i}$, and use equation
(\ref{eq:CDdecomp}) to obtain Eve's precoder $\hat{\mathbf{V}}_{i,E}$.

\subsection{Received signals}
The received signals at Alice, Bob and Eve can be given as 
\begin{align}
	\mathbf{y}_{A} & =\mathbf{H}_{AB}\mathbf{x}_{B}+\mathbf{G}_{A}\mathbf{x}_{A}+\mathbf{z}_{A},\\
	\mathbf{y}_{B} & =\mathbf{G}_{B}\mathbf{x}_{B}+\mathbf{H}_{BA}\mathbf{x}_{A}+\mathbf{z}_{B},\\
	\mathbf{y}_{E} & =\mathbf{H}_{EA}\mathbf{x}_{A}+\mathbf{H}_{EB}\mathbf{x}_{B}+\mathbf{z}_{E} =\underbrace{\left[\mathbf{H}_{EA},\mathbf{H}_{EB}\right]}_{\mathbf{H}_{E}}\underbrace{\left[\begin{array}{c}
			\mathbf{x}_{A}\\
			\mathbf{x}_{B}
		\end{array}\right]}_{\mathbf{x}_E}+\mathbf{z}_{E},
\end{align}
where noise vectors are Gaussian $\mathbf{z}_{i}\sim\mathcal{CN}\left(\mathbf{0},\sigma^{2}\mathbf{I}\right),\forall i=A,B,E$;
and
\begin{equation}
	\mathbf{x}_E =\mathbf{V}_E \left(\mathbf{P}_{s,E}^{\frac{1}{2}}\mathbf{s}_E+\mathbf{P}_{w,E}^{\frac{1}{2}}\mathbf{w}_E \right)+\underline{\mathbf{V}}_E \mathbf{P}_{\underline{w},E}^{\frac{1}{2}}\underline{\mathbf{w}}_E ,
\end{equation}
with $\mathbf{V}_E =\mathcal{BD}\left(\mathbf{V}_{A},\mathbf{V}_{B}\right)$,
$\mathbf{P}_{v, E}=\mathcal{BD}\left(\mathbf{P}_{v,A},\mathbf{P}_{v,B}\right)$,
$v \in\left\{ s,w,\underline{w}\right\} $ and $\mathbf{v}^{T}=\left[\mathbf{v}_{A}^{T},\mathbf{v}_{B}^{T}\right]$, where
$\mathbf{v}\in\left\{ \mathbf{s},\mathbf{w},\underline{\mathbf{w}}\right\} $.
The covariance matrix of $\mathbf{x}_E$ can be written as $\mathbf{T}_E =\mathcal{BD}\left(\mathbf{T}_{A},\mathbf{T}_{B}\right),$
since the transmitted signals from Alice and Bob ($\mathbf{x}_{A}$
and $\mathbf{x}_{B}$) are uncorrelated, that is, $\mathbb{E}\left\{ \mathbf{x}_{A}\mathbf{x}_{B}^{H}\right\} =\mathbf{0}.$

For the received signals, we define the covariance matrices as 
\begin{align}
	& \mathbf{R}_{y,A}=\mathbb{E}\left\{ \mathbf{y}_{A}\mathbf{y}_{A}^{H}\right\} =\hat{\mathbf{H}}_{AB}\mathbf{T}_{B}\hat{\mathbf{H}}_{AB}^{H}+\mathbf{G}_{A}\mathbf{T}_{A}\mathbf{G}_{A}^{H} \nonumber \\
	& \;\;\;\;\;\;\;\;\;\;\;\;\;\;\;\;\;\;\;\;\;\;\;\;\;\;\;\;\;\;\;\;\;\;\;\; +\left(\sigma_{\Delta,AB}^{2}tr(\mathbf{T}_{B})+\sigma^{2}\right)\mathbf{I}, \\
	& \mathbf{R}_{y,B}=\mathbb{E}\left\{ \mathbf{y}_{B}\mathbf{y}_{B}^{H}\right\}  =\mathbf{G}_{B}\mathbf{T}_{B}\mathbf{G}_{B}^{H}+\hat{\mathbf{H}}_{BA}\mathbf{T}_{A}\hat{\mathbf{H}}_{BA}^{H} \nonumber \\
	& \;\;\;\;\;\;\;\;\;\;\;\;\;\;\;\;\;\;\;\;\;\;\;\;\;\;\;\;\;\;\;\;\;\;\;\; +\left(\sigma_{\Delta,BA}^{2}tr(\mathbf{T}_{A})+\sigma^{2}\right)\mathbf{I}, \\
	& \mathbf{R}_{y,E}=\mathbb{E}\left\{ \mathbf{y}_{E}\mathbf{y}_{E}^{H}\right\} =\mathbf{H}_{E}\mathbf{T}_E\mathbf{H}_{E}^{H}+\sigma^{2}\mathbf{I}.
\end{align}

\subsection{Secrecy rates \label{subsec:secrecy_rate}}
MMSE receiver with successive interference cancellation can achieve the capacity of the MIMO channel
\cite{tse2005fundamentals}, which refers to the maximum achievable information transmission rate. Nevertheless, maximum likelihood-based symbol detection can be applied later to provide the optimal data recovery. For Alice and Bob, they apply the three-stage SIC first and then employ MMSE combing to he received signal, and their information rates can be given as 
\begin{align}
	& R_{BA}=\log_{2}\left|\mathbf{I}+\hat{\mathbf{H}}_{BA}\mathbf{V}_{A}\mathbf{P}_{s,A}\mathbf{V}_{A}^{H}\hat{\mathbf{H}}_{BA}^{H}\mathbf{C}_{BA}^{-1}\right|,\\
	& R_{AB}=\log_{2}\left|\mathbf{I}+\hat{\mathbf{H}}_{AB}\mathbf{V}_{B}\mathbf{P}_{s,B}\mathbf{V}_{B}^{H}\hat{\mathbf{H}}_{AB}^{H}\mathbf{C}_{AB}^{-1}\right|,
\end{align}
where covariance matrices can be written as 
\begin{align}
	& \mathbf{C}_{BA}=\mathbf{G}_{B}\mathbf{T}_{B}\mathbf{G}_{B}^{H}+\theta\hat{\mathbf{H}}_{BA}\underline{\mathbf{V}}_{A}\mathbf{P}_{\underline{w},A}\underline{\mathbf{V}}_{A}^{H}\hat{\mathbf{H}}_{BA}^{H} \nonumber \\
	& \qquad \quad +\theta\hat{\mathbf{H}}_{BA}\mathbf{V}_{A}\mathbf{P}_{w,A}\mathbf{V}_{A}^{H}\hat{\mathbf{H}}_{BA}^{H}+\left(\sigma_{\Delta,BA}^{2}P_{A}+\sigma^{2}\right)\mathbf{I},\\
	& \mathbf{C}_{AB}=\mathbf{G}_{A}\mathbf{T}_{A}\mathbf{G}_{A}^{H}+\theta\hat{\mathbf{H}}_{AB}\underline{\mathbf{V}}_{B}\mathbf{P}_{\underline{w},B}\underline{\mathbf{V}}_{B}^{H}\hat{\mathbf{H}}_{AB}^{H} \nonumber \\
	& \qquad \quad +\theta\hat{\mathbf{H}}_{AB}\mathbf{V}_{B}\mathbf{P}_{w,B}\mathbf{V}_{B}^{H}\hat{\mathbf{H}}_{AB}^{H}+\left(\sigma_{\Delta,AB}^{2}P_{B}+\sigma^{2}\right)\mathbf{I}.
\end{align}
The variable
$\theta$ is introduced to represent the cases when the artificial
noise is known ($\theta=0$), and unknown ($\theta=1$) between Alice
and Bob. For example, if artificial noise transmitted from Alice is
known at Bob, Bob can subtract the same to get better data estimates.
For Eve, she receives signals from both Alice and Bob due to the IBFD operation. We assume Eve is closer to Alice without loss of generality as the one that is closer to the eavesdropper can be regarded as Alice in practice. Thus, Eve processes the signal from Alice first, and then cancels the interference from Alice and processes the signal from Bob. Based on the successive cancellation approach at the receiver, the Eve's rate can be divided into two parts $R_{E}=R_{EA}+R_{EB}$
with 
\begin{align}
	& R_{EA}=\log_{2}\left|\mathbf{I}+\mathbf{H}_{EA}\hat{\mathbf{V}}_{A,E}\mathbf{P}_{s,A}\hat{\mathbf{V}}_{A,E}^{H}\mathbf{H}_{EA}^{H}\mathbf{C}_{E}^{-1}\right|,\label{eq:REA}\\
	& R_{EB}=\log_{2}\Big|\mathbf{I}+\mathbf{H}_{EB}\hat{\mathbf{V}}_{B,E}\mathbf{P}_{s,B}\hat{\mathbf{V}}_{B,E}^{H}\mathbf{H}_{EB}^{H}\\
	& \qquad\times\left(\mathbf{H}_{EA}\hat{\mathbf{V}}_{A,E}\mathbf{P}_{s,A}\hat{\mathbf{V}}_{A,E}^{H}\mathbf{H}_{EA}^{H}+\mathbf{C}_{E}\right)^{-1}\Big|,
\end{align}
where the covariance matrix is given as
\begin{align}
     \mathbf{C}_{E}=\mathbf{H}_{E}\left(\mathbf{T}_E-\hat{\mathbf{V}}_{E}\mathbf{P}_{s,E}\hat{\mathbf{V}}_{E}^{H}\right)\mathbf{H}_{E}^{H}+\sigma^{2}\mathbf{I}.
\end{align}
It can be noted that for Eve, the available precoder $\hat{\mathbf{V}}_{E}$
is substituted for the actual precoder $\mathbf{V}$. 

For a given CSI instance, the true secrecy rates can be defined as
\begin{align}
	R_{Si}=\max\left \{ 0,R_{ji}-R_{Ei} \right \} ,
\end{align}
where $i\in\left\{ A,B\right\}$ and $j\in\left\{ A,B\right\} \setminus\left\{ i\right\} $. IBFD radios have the potential to improve the secrecy rate in two factors.
\begin{itemize}
	\item IBFD radios can enable simultaneous probing between legitimate users so that they can harness a better channel reciprocity, which could be reflected by a smaller channel uncertainty.
	\item Eve will be affected by each other when receiving transmissions from Alice and Bob at the same time, so as to reduce its achievable eavesdropping data rate.
\end{itemize}

\section{Two-step secrecy power allocation \label{sec:Two-step-secrecy-power}}
In this section, we derive the power allocation for Alice and
Bob to maximize sum secrecy rates. However, it should be noted that Alice and Bob are not able to calculate the secrecy rate due to the lack of Eve's knowledge, e.g., CSI of the eavesdropping channel, the precoder Eve utilized for calculation, etc. To derive the optimal power allocation policy, we assume legitimate users are aware of the presence of Eve and can acquire her location, so they have statistical information about the eavesdropping channel, i.e., path loss $\beta_{EA}$ and $\beta_{EB}$. This assumption is feasible with the development of integrated sensing and communication (ISAC)\cite{9847100}, which gives the terminal devices the human-like cognitive ability to construct a dynamic topology of the surrounding environment as in \cite{9557830}. We can reasonably assume that legitimate users are aware of their surroundings and can therefore locate suspected eavesdroppers from the constructed dynamic topology. In contrast, we do not limit Eve's ability and assume she has the knowledge of $\hat{\mathbf{H}}_{AB}$
with uncertainty. Therefore, we solve the above in two steps. First,
we obtain the coarse power allocation based on ergodic rate approximations.
Subsequently, the scalar power values are distributed for matrix power
allocation based on the small scale fading $\hat{\mathbf{H}}_{BA}$
and $\hat{\mathbf{H}}_{AB}$, that is, a fine power allocation step.
It can be noted that Eve is only aware of coarse power values, since
they are based on large-scale fading components and can be computed
by Eve. 

\subsection{Rate approximation}
The optimization problem in order to maximize sum ergodic secrecy
rates subject to transmit power constraints at Alice and Bob can be
cast as 
\begin{align}
	\label{eq:optSR}
	& \max_{\mathbf{P}_{v,i},\forall v,i}  \mathbb{E}\left\{ R_{BA}-R_{EA}+R_{AB}-R_{EB}\right\} \\
	\text{s.t.} ~& tr\left(\mathbf{P}_{s,i}+\mathbf{P}_{w,i}\right)+tr\left(\mathbf{P}_{\underline{w},i}\right)\leq P_{i},\forall i=A,B,\nonumber 
\end{align}
where $v\in\left\{ s,w,\underline{w}\right\} $ and $i\in\left\{ A,B\right\} $. Note that we aim to maximize the sum secrecy rates without restricting the individual secrecy rate to be equal or greater than 0 (i.e., $R_{sA} = R_{BA} - R_{EA}\geq 0$ and  $R_{sB} = R_{AB} - R_{EB}\geq 0$) to obtain the optimal solution. For instance, there could be a situation where there is a power allocation policy that can make Alice's secrecy rate less than 0 but the sum secrecy rates are greater than the allocation policy that can make Alice's secrecy rate equal or greater than 0. From the system's perspective, the first power allocation policy is better because it allows for a higher secure transmission rate. The above problem is difficult to solve due to non-convexity\footnote{Fractions inside and difference of the logarithms are not convex.}
and complex matrix differentiations. Instead, we approximate the ergodic rate employing the following result. 
\begin{lemma}
	\label{lem:MMSE-rate}For the received signal equation $\mathbf{y}=\mathbf{H}\mathbf{x}+\mathbf{z},$
	where $\mathbf{z}\sim\mathcal{CN}\left(\mathbf{0},\mathbf{C}_{z}\right)$
	and $\mathbf{x}\sim\mathcal{CN}\left(\mathbf{0},\mathbf{C}_{x}\right)$,
	the rate approximation can be obtained as 
	\begin{equation}
		\mathbb{E}\left \{ R \right \}\approx\log_{2}\left|\mathbf{I}+\mathbf{C}_{x}\mathbb{E}\left\{ \mathbf{H}^{H}\mathbb{E}\left\{ \mathbf{C}_{z}\right\} ^{-1}\mathbf{H}\right\} \right|.\label{eq:lem-eqn}
	\end{equation}
\end{lemma}
\begin{proof}
	Proof is given in  Appendix \ref{subsec:Proof-of-Lemma1}.
\end{proof}

\begin{corollary}
	\label{cor:ergodicrate}Based on the results given in Lemma \ref{lem:MMSE-rate}, ergodic rate approximations of Alice, Bob and Eve can be given as 
	\begin{align}
		& \mathbb{E}\left \{ R_{BA} \right \}\approx\tilde{R}_{BA}=b\log_{2}\left|1+\frac{N_{B}\beta_{BA}P_{s,A}}{b c_{BA}}\right|,\\
		& \mathbb{E} \left \{ R_{AB} \right \}\approx\tilde{R}_{AB}=b\log_{2}\left|1+\frac{N_{A}\beta_{AB}P_{s,B}}{b c_{AB}}\right|,\\
		& \mathbb{E}\left \{ R_{EA} \right \}\approx\tilde{R}_{EA}=b\log_{2}\left|1+\frac{N_{E}\beta_{EA}P_{s,A}}{b c_{E}}\right|,\label{eq:RE-appr}\\
		& \mathbb{E}\left \{ R_{EB} \right \}\approx\tilde{R}_{EB}=b\log_{2}\left|1+\frac{N_{E}\beta_{EB}P_{s,B}}{b c_{E}}\right|,
	\end{align}
	where
	\begin{align}
		P_{v,i}=tr(\mathbf{P}_{v,i}),~ v\in\left\{ s,w,\underline{w}\right\} ,i\in\left\{ A,B\right\} ,
	\end{align}
	\begin{align}
		\label{eq:cBA}
		c_{BA}=&tr(\mathbf{T}_{B})\eta+\theta\beta_{BA}\left(P_{w,A}+P_{\underline{w},A}\right) \nonumber \\
		&+\sigma_{\Delta,BA}^{2}tr(\mathbf{T}_{A})+\sigma^{2} ,
	\end{align}
	\begin{align}
		\label{eq:cAB}
		c_{AB}=&tr(\mathbf{T}_{A})\eta+\theta\beta_{AB}\left(P_{w,B}+P_{\underline{w},B}\right) \nonumber \\
		& +\sigma_{\Delta,AB}^{2}tr(\mathbf{T}_{B})+\sigma^{2} ,
	\end{align}
	\begin{align}
		\label{eq:cE}
		c_{E}=&\beta_{EA}\left(P_{s,A}\kappa_{A,E}+P_{w,A}+P_{\underline{w},A}\right) \nonumber \\
		&+\beta_{EB}\left(P_{s,B}\kappa_{B,E}+P_{w,B}+P_{\underline{w},B}\right)+\sigma^{2} .
	\end{align}
\end{corollary}
\begin{proof}
	Proof is given in Appendix \ref{subsec:Proof-of-Corollary2}.
\end{proof}

\begin{figure*}
	\centering
	\subfigure[Bob and Eve]{
		\begin{minipage}[t]{0.5\linewidth}
			\centering
			\includegraphics[width=1.0\linewidth]{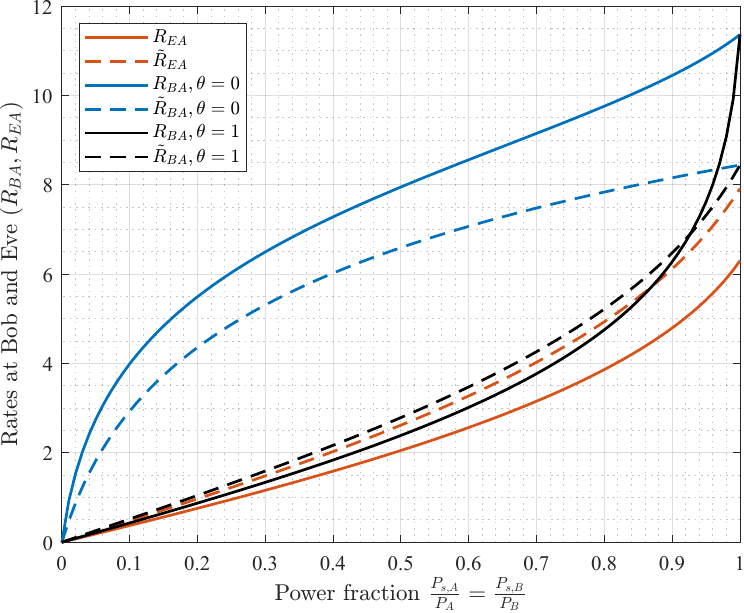}
		\end{minipage}%
	}%
	\subfigure[Alice and Eve]{
		\begin{minipage}[t]{0.5\linewidth}
			\centering
			\includegraphics[width=1.0\linewidth]{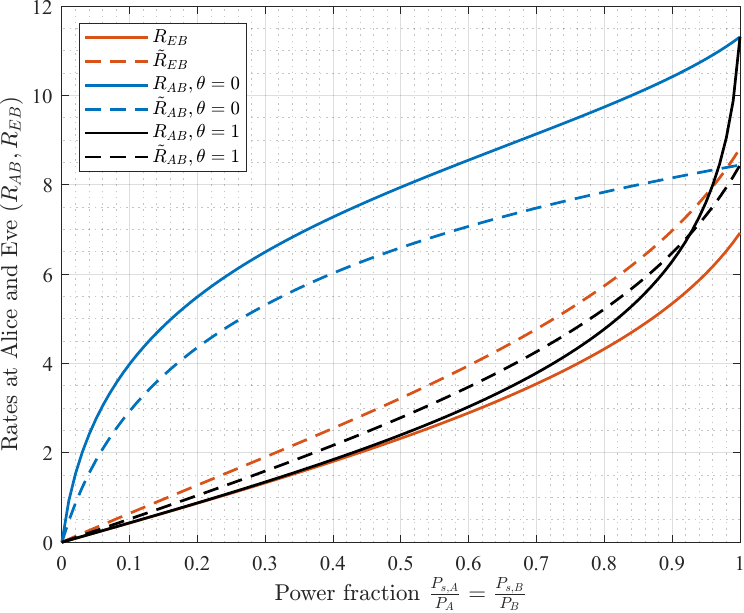} 
		\end{minipage}%
	}%
	\centering
	\caption{Rates at Alice $(0,0)$, Bob $(0,1)$ and Eve $(0.5,5)$ with approximations for $P_{s,i}=1-P_{w,i}-P_{\underline{w},i}$, $P_{w,i}=P_{\underline{w},i}$, $P_{i}=25$ dB, $\kappa_{iE}=0.1$, $N_{A}=N_{B}=\frac{N_{E}}{2}=2b=4$.\label{fig:appr}}
\end{figure*}
Figure \ref{fig:appr} plots the rates for the transmit signals from
Alice (left) and Bob (right), averaged over 100 channel realizations.
It can be observed that the approximation is close at lower fraction
($\frac{P_{s,i}}{P_{i}}$), and the gap becomes larger at higher
values. The trend of plots is useful in solving for optimum power
allocation. For $\theta=1$, curves show convex behavior, whereas, for $\theta=0$, a concave function can be approximated.

\subsection{Coarse power allocation}
Based on the approximations,
the above problem in Equation \eqref{eq:optSR} can be recast as 
\begin{align}
	& \max_{P_{v,i},\forall v,i} \tilde{R}_{BA}-\tilde{R}_{EA}+\tilde{R}_{AB}-\tilde{R}_{EB},
	\\
	\text{s.t.} ~& tr\left(\mathbf{P}_{s,i}+\mathbf{P}_{w,i}\right)+tr\left(\mathbf{P}_{\underline{w},i}\right)\leq P_{i},\forall i=A,B. \nonumber 
\end{align}
Since contributions from $P_{w,i}$ and $P_{\underline{w},i}$ are
similar in approximations, these variables can be aggregated. Moreover,
the constraints are considered to be satisfied with equality since more strength can be allocated to artificial
noise to reduce Eve rates. Let $\gamma_{i}=\frac{P_{s,i}}{P_{i}}\in\left[0,1\right],i\in\left\{ A,B\right\} $
be the fraction of power budget allocation in the desired signal direction.
Consequently, the aggregated artificial noise power can be obtained
as $P_{w,i}+P_{\underline{w},i}=\left(1-\gamma_{A}\right)P_{A}$.
Substituting this, the optimization problem can be reduced as 
\begin{equation}
	\max_{0\leq\gamma_{i}\leq1,i\in\left\{ A,B\right\} }\tilde{R}_{BA}-\tilde{R}_{EA}+\tilde{R}_{AB}-\tilde{R}_{EB}.
\end{equation}
The solution can be obtained using successive convex approximation
algorithm. However, to describe and verify the solution in different
cases, the positive secrecy regions are analyzed as follows. Using
epigraph technique, one can rewrite the optimization problem as 
\begin{align}
    \max_{0\leq\gamma_{i}\leq1,t_{i},i\in\left\{ A,B\right\} } t_{A} & +t_{B} \\
	\text{s.t.} ~ \tilde{R}_{BA}-\tilde{R}_{EA} &\geq t_{A},\label{eq:RBAcon} \\
	\tilde{R}_{BA}-\tilde{R}_{EA} &\geq t_{B},\label{eq:RABcon}\\
	t_{A},t_{B} & \geq0.
\end{align}
Assuming $\mathbf{u}^{T}=\left[\gamma_{A},\gamma_{B}\right]$, Equations \eqref{eq:cBA}-\eqref{eq:cE} can be simplified as in Appendix \ref{subsec:Quadratic-cons-equation}, providing
\begin{align}
	c_{BA} &=a_{BA}\mathbf{e}_{1}^{T}\mathbf{u}+g_{BA}, \\
	c_{AB} & =a_{AB}\mathbf{e}_{2}^{T}\mathbf{u}+g_{AB}, \\
	c_{E} & =\mathbf{a}_{E}^{T}\mathbf{u}+g_{E} .
\end{align}
Accordingly, the constraints in (\ref{eq:RBAcon})-(\ref{eq:RABcon}) can be rearranged
for $\bar{t}_{i}=2^{\frac{t_{i}}{d}},i\in\left\{ A,B\right\}$ as the equations on the top of the next page, which depicts bi-variate quadratic equations with $\bar{t}_i = 2^{t_i}$. 
\newcounter{mytempeqncnt}
\begin{figure*}
    \normalsize
    \setcounter{mytempeqncnt}{\value{equation}}
    \begin{equation}
    \begin{split}
       	    \implies & \mathbf{u}^{T}\left[\frac{\mathbf{a}_{E}\mathbf{e}_{1}^{T}+\mathbf{e}_{1}\mathbf{a}_{E}^{T}}{2}\left(a_{BA}(1-\bar{t}_{A})+\frac{N_{B}\beta_{BA}P_{A}}{d}\right)-\bar{t}_{A}a_{BA}\frac{N_{E}\beta_{EA}P_{A}}{d}\mathbf{e}_{1}\mathbf{e}_{1}^{T}\right]\mathbf{u}\nonumber \\
		& \quad +\mathbf{u}^{T}\left[\left(\mathbf{a}_{E}b_{BA}+\mathbf{e}_{1}a_{BA}b_{E}\right)(1-\bar{t}_{A})+\mathbf{e}_{1}b_{E}\frac{N_{B}\beta_{BA}P_{A}}{d}-\mathbf{e}_{1}\bar{t}_{A}b_{BA}\frac{N_{E}\beta_{EA}P_{A}}{d}\right] \nonumber \\
		& \quad +b_{BA}b_{E}(1-\bar{t}_{A})\geq0, \\
	    \implies & \mathbf{u}^{T}\left[\frac{\bar{\mathbf{a}}_{E}\mathbf{e}_{2}^{T}+\mathbf{e}_{2}\bar{\mathbf{a}}_{E}^{T}}{2}\left(a_{AB}(1-\bar{t}_{B})+\frac{N_{A}\beta_{AB}P_{B}}{d}\right)-\bar{t}_{B}a_{AB}\frac{N_{E}\beta_{EB}P_{B}}{d}\mathbf{e}_{2}\mathbf{e}_{2}^{T}\right]\mathbf{u}\nonumber \\
		& \quad +\mathbf{u}^{T}\left[\left(\bar{\mathbf{a}}_{E}b_{AB}+\mathbf{e}_{2}a_{AB}b_{E}\right)(1-\bar{t}_{B})+\mathbf{e}_{2}b_{E}\frac{N_{A}\beta_{AB}P_{B}}{d}-\mathbf{e}_{2}\bar{t}_{B}b_{AB}\frac{N_{E}\beta_{EB}P_{B}}{d}\right] \nonumber \\
		& \quad +b_{AB}b_{E}(1-\bar{t}_{B})\geq0,
    \end{split}
    \end{equation}
    \hrulefill
    \vspace*{4pt}
\end{figure*}

\subsubsection{When Eve is far away or not present}
This scenarios suggests $\beta_{Ei}=0$ and $R_{E}=0$. In this scenario, power can be allocated according to
the constraints $t_{i}$, that is, solving $R_{BA}\geq t_{A}$ and
$R_{AB}\geq t_{B}$ as
\begin{align*}
	& 1+\frac{N_{j}\beta_{ji}\gamma_{i}P_{i}}{bc_{ji}}\geq\bar{t}_{i}\\
	\implies & \left(a_{ji}\gamma_{i}+g_{ji}\right)(1-\bar{t}_{i})+\frac{N_{j}\beta_{ji}\gamma_{i}P_{i}}{b}\geq0\\
	\implies & \gamma_{i}\left(a_{ji}(1-\bar{t}_{i})+\frac{N_{j}\beta_{ji}P_{i}}{b}\right)+g_{ji}(1-\bar{t}_{i})\geq0\\
	\implies & \gamma_{i}=-\frac{g_{ji}(1-\bar{t}_{i})}{a_{ji}(1-\bar{t}_{i})+\frac{N_{j}\beta_{ji}P_{i}}{b}}\geq0 ,\\
\end{align*}
providing the solutions as
\begin{align}
	\gamma_{A} & =\frac{g_{BA}(\bar{t}_{A}-1)}{a_{BA}(1-\bar{t}_{A})+\frac{N_{B}\beta_{BA}P_{A}}{b}},\\
	\gamma_{B} & =\frac{g_{AB}(\bar{t}_{B}-1)}{a_{AB}(1-\bar{t}_{B})+\frac{N_{A}\beta_{AB}P_{B}}{b}}.
\end{align}
The solutions rely on parameters that are available at Alice and Bob, so it is feasible.

\subsubsection{When secrecy rates are non-negative}
Let $t_{i}=0,i\in\left\{ A,B\right\} $ such that $\bar{t}_{i}=2^{t_i}=1$. Then, the bi-variate quadratic equations are reduced
to line equations as 
\begin{align}
	& \mathbf{u}^{T}\left[\mathbf{a}_{E}\left(\frac{N_{B}\beta_{BA}P_{A}}{b}\right)-a_{BA}\frac{N_{E}\beta_{EA}P_{A}}{b}\mathbf{e}_{1}\right] \nonumber \\
	&+\left[g_{E}\frac{N_{B}\beta_{BA}P_{A}}{b}-g_{BA}\frac{N_{E}\beta_{EA}P_{A}}{b}\right]\geq0,\\
	& \mathbf{u}^{T}\left[\mathbf{a}_{E}\left(\frac{N_{A}\beta_{AB}P_{B}}{b}\right)-a_{AB}\frac{N_{E}\beta_{EB}P_{B}}{b}\mathbf{e}_{2}\right] \nonumber \\
	&+ \left[g_{E}\frac{N_{A}\beta_{AB}P_{B}}{b}-g_{AB}\frac{N_{E}\beta_{EB}P_{B}}{b}\right]\geq0 .
\end{align}
\begin{figure*}
	\centering
	\subfigure[positive secrecy region at Bob with known AN]{
		\begin{minipage}[t]{0.5\linewidth}
			\centering
			\includegraphics[width=1.0\linewidth]{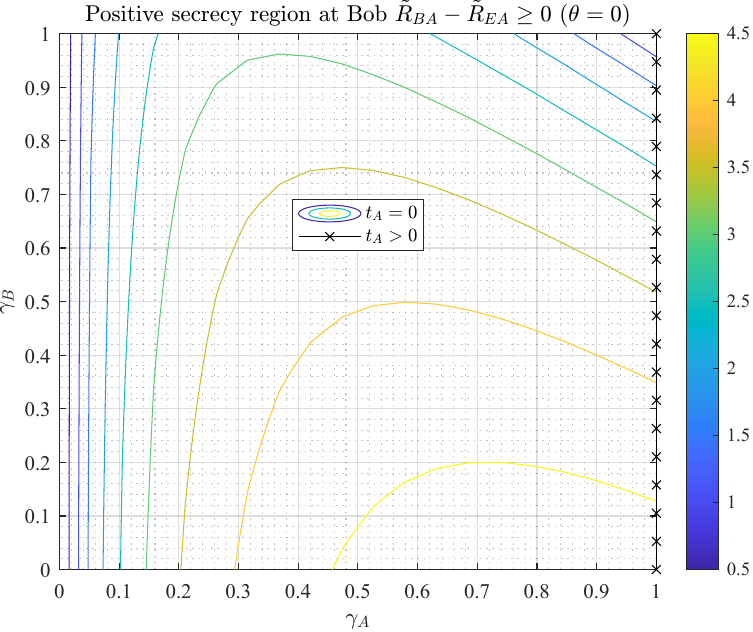}
		\end{minipage}%
	}%
	\subfigure[positive secrecy region at Alice with known AN]{
		\begin{minipage}[t]{0.5\linewidth}
			\centering
			\includegraphics[width=1.0\linewidth]{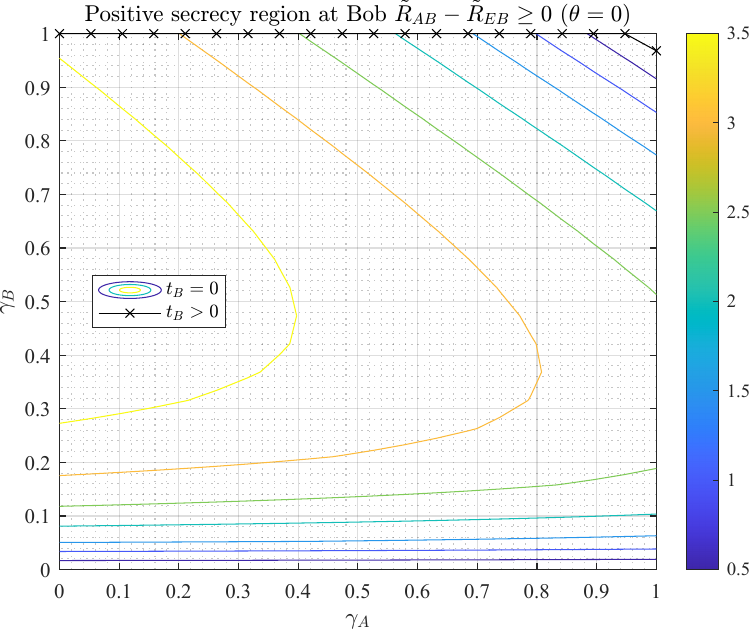} 
		\end{minipage}%
    }%
    
	\subfigure[positive secrecy region at Bob with unknown AN]{	
		\begin{minipage}[t]{0.5\linewidth}
			\centering
			\includegraphics[width=1.0\linewidth]{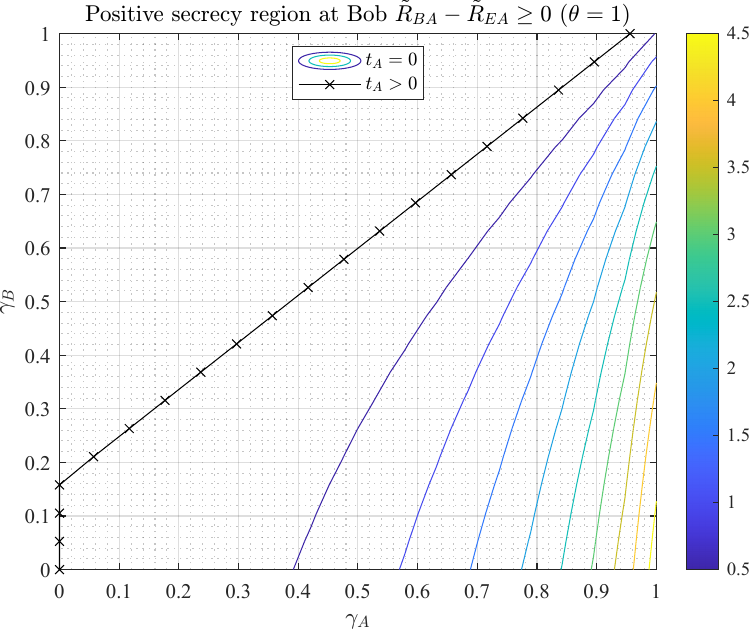}
		\end{minipage}%
		\label{fig:s-region3}
	}%
		\subfigure[positive secrecy region at Alice with unknown AN]{
		\begin{minipage}[t]{0.5\linewidth}
			\centering
			\includegraphics[width=1.0\linewidth]{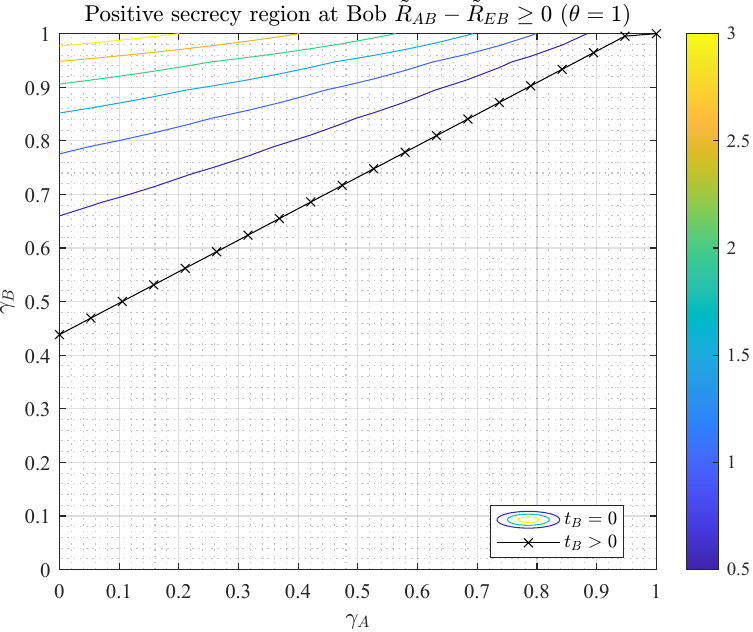} 
		\end{minipage}%
		\label{fig:s-region4}
	}%
	\centering
	\caption{Secrecy rate regions for Alice and Bob streams with known and unknown artificial noise $\theta=0,1$.\label{fig:s-region}}
\end{figure*}
To visualize above equations, Figure \ref{fig:s-region} depicts contour
plots for secrecy regions involving linear and quadratic equations
with $\theta=0$ (known AN) and $\theta=1$ (unknown AN), and simulation settings are same as in Figure \ref{fig:appr}. It can be
observed that for $\theta=0$, left (or below) of $t_{i}=0$ line,
secrecy rates are positive for Alice (or Bob) rates, whereas below
(or above) is the positive region when $\theta=1$. In other words,
for $\theta=1$, higher secrecy rate values are concentrated towards,
higher $\gamma_{A}$ (Figure \ref{fig:s-region3}) and $\gamma_{B}$ (Figure \ref{fig:s-region4}).
It can be inferred that for $\theta=1$, due to monotonic rates, obtaining
optimum solution for $\mathbf{u}$ is simpler than the case of $\theta=0$.
Therefore, for $\theta=1$, we employ first-order derivative based
iterative successive convex approximation algorithm, which is similar
to gradient descent method. On the other hand, for $\theta=0$, Hessian
based updates are used, also known as Newton seconder order method.
For completeness, steps are summarized for both $\theta=0$ and $\theta=1$
in the Algorithm \ref{alg:A1S1}, where 
\begin{equation}
	f(\mathbf{u})=\tilde{R}_{BA}-\tilde{R}_{EA}+\tilde{R}_{AB}-\tilde{R}_{EB}.
\end{equation}
In this procedure, after initializing the variables $\mathbf{u}$,
they are updated according to the update equation until the convergence,
that is, $\left\Vert \nabla f\right\Vert _{2}\leq\epsilon$, where
$\epsilon>0$ is a small tolerance value. In general, entries of gradient
and Hessian for each rate expression with respect to a scalar variable
(say $x$), say $\mathcal{F}(x)=\log\left|\mathbf{I}+\mathbf{F}_{i}(x)\mathbf{G}_{i}^{-1}(x)\right|$,
can be obtained as 
\begin{align}
	& \frac{\partial \mathcal{F}}{\partial x}=tr\left[\left(\mathbf{G}_{i}+\mathbf{F}_{i}\right)^{-1}\left(-\mathbf{F}_{i}\mathbf{G}_{i}^{-1}\frac{\partial\mathbf{G}_{i}}{\partial x}+\frac{\partial\mathbf{F}_{i}}{\partial x}\right)\right],\label{eq:fdiff}\\
	& \frac{\partial^{2}\mathcal{F}}{\partial y\partial x}=tr\Big[\left(\mathbf{G}_{i}+\mathbf{F}_{i}\right)^{-1}\Big\{-\frac{\partial\mathbf{G}_{i}+\mathbf{F}_{i}}{\partial y}\left(\mathbf{G}_{i}+\mathbf{F}_{i}\right)^{-1}\nonumber \\
	& \times\left(-\mathbf{F}_{i}\mathbf{G}_{i}^{-1}\frac{\partial\mathbf{G}_{i}}{\partial x}+\frac{\partial\mathbf{F}_{i}}{\partial x}\right)-\frac{\partial\mathbf{F}_{i}}{\partial y}\mathbf{G}_{i}^{-1}\frac{\partial\mathbf{G}_{i}}{\partial x}\nonumber \\
	& +\mathbf{F}_{i}\mathbf{G}_{i}^{-1}\frac{\partial\mathbf{G}_{i}}{\partial y}\mathbf{G}_{i}^{-1}\frac{\partial\mathbf{G}_{i}}{\partial x}-\mathbf{F}_{i}\mathbf{G}_{i}^{-1}\frac{\partial^{2}\mathbf{G}_{i}}{\partial y\partial x}+\frac{\partial^{2}\mathbf{F}_{i}}{\partial y\partial x}\Big\}\Big].
\end{align}
Regarding the convergence, since the problem contains two variables,
gradient descent method converge in few iterations (of the order of
$10$); especially Newton method achieve the optimum value in less
than $10$ iterations for two variables. Hence, for these popular
methods, plots for optimization progress versus number of iterations
are omitted.
\begin{algorithm}
	\begin{algorithmic}[1]
		
		\State{\label{A1s1}Initialize $\mathbf{u}\in\left[0,1\right]^{2}$}
		
		\For{ $t=1,2,\ldots$}
		
		\State{compute the gradient $\nabla f$ and Hessian $\nabla\nabla^{H}f$,
		}
		
		\If{$\theta=1$}
		
		\State{$\mathbf{u}\leftarrow\mathbf{u}-\left[\nabla\nabla^{H}f\right]^{-1}\nabla f$,}
		
		\Else
		
		\State{$\mathbf{u}\leftarrow\mathbf{u}-\nabla f$,}
		
		\State{clip $\mathbf{u}$ within box constraints $\mathbf{u}\in\left[0,1\right]^{2}$,}
		
		\EndIf
		
		\State{break at convergence, i.e., $\left\Vert \nabla f\right\Vert _{2}\leq\epsilon$,
		}
		
		\EndFor
		
	\end{algorithmic}
	
	\caption{Secrecy power allocation using SCA algorithm.\label{alg:A1S1}}
\end{algorithm}

\subsubsection{High SNR regime\label{subsec:High-SNR-regime}}
Secrecy region plots show that the case of $\gamma_{i}^{*}=1$ occurs
mostly when $\theta=1$ (see yellow contour lines). It is due to the
fact that the product in the rates (at Bob) $\theta\beta_{BA}\left(P_{w,A}+P_{\underline{w},A}\right)=\theta\beta_{BA}\left(1-\gamma_{A}\right)P_{A}$
with $\theta=1$ can be minimized when $\gamma_{A}=1$. 
When $P_{i}=P\rightarrow\infty$, the rates can be given as 
\begin{align}
	& \tilde{R}_{BA}=b\log_{2}\frac{N_{B}\beta_{BA}\gamma_{A}}{b\left(\eta+\theta\beta_{BA}\left(1-\gamma_{A}\right)+\sigma_{\Delta,BA}^{2}\right)},\label{eq:highSNR_RBA}\\
	& \tilde{R}_{EA}=b\log_{2}\frac{N_{E}\beta_{EA}\gamma_{A}}{b\sum_{i\in\left\{ A,B\right\} }\beta_{Ei}\left(\gamma_{i}\kappa_{i,E}+1-\gamma_{i}\right)}.\label{eq:highSNR_REA}
\end{align}
With $\gamma_{i}=1$, we have 
\begin{align}
	& \tilde{R}_{BA}=b\log_{2}\frac{N_{B}\beta_{BA}}{b\left(\eta+\sigma_{\Delta,BA}^{2}\right)},\\
	& \tilde{R}_{EA}=b\log_{2}\frac{N_{E}\beta_{EA}}{b\left(\beta_{EA}\kappa_{A,E}+\beta_{EB}\kappa_{B,E}\right)},
\end{align}
which shows that if Eve has a higher number of antennas (or a lower path
loss, or more information about the legitimate channel), it provides $\tilde{R}_{EA}>\tilde{R}_{BA}$,
leading to no-secrecy rates. The term $\gamma_{i}=1$ allocates all
power to desired signals and no-power to artificial noise, which makes
Eve easier to decode signals.

\subsection{Fine power allocation}
Since the accurate CSI $\hat{\mathbf{H}}_{AB}$ is only known to Alice
and Bob, the fine power allocation step can help improve the secrecy.
As Alice and Bob are unaware of the eavesdropping channels $\mathbf{H}_{EA}$
and $\mathbf{H}_{EB}$, Alice and Bob can only tailor this design
to maximize their own rates. The optimization problem can be cast
as 
\begin{align}
    & \max_{\mathbf{P}_{s,i},\mathbf{P}_{w,i},\mathbf{P}_{\underline{w},i},i\in\left\{ A,B\right\} } R_{AB}+R_{BA}\\
	\text{s.t. }~ tr & \left ( \mathbf{P}_{s,i} \right )=\gamma_{i}P_{i},\\
	tr & \left(\mathbf{P}_{\underline{w},i}+\mathbf{P}_{w,i}\right)=\left(1-\gamma_{i}\right)P_{i},i\in\left\{ A,B\right\} ,
\end{align}
which is difficult to solve due to matrix fractions, and we give suboptimal solutions as follows.

\subsubsection{Equal power allocation}
For equal power allocation, we set 
\begin{gather}
	\mathbf{P}_{s,i}=\gamma_{i}P_{i}\frac{1}{b}\mathbf{I}, ~ \mathbf{P}_{w,i}=\frac{1-\gamma_{i}}{2}P_{i}\frac{1}{b}\mathbf{I},~\textup{and}~\mathbf{P}_{\underline{w},i}=\frac{1-\gamma_{i}}{2}P_{i}\frac{1}{n_i}\mathbf{I}.\label{eq:fine_equal}
\end{gather}
\subsubsection{Suboptimal power allocation}
For a sub-optimal solution, we decouple the problem for Alice and
Bob, and choose $\mathbf{P}_{s,i}$ to maximize the desired signal
power, and $\mathbf{P}_{\underline{w},i}$ to minimize the RSI plus
noise power. For $\mathbf{P}_{s,i}$, we write 
\begin{align}  
    \arg\max_{\mathbf{P}_{s,i}} & tr  \left(\hat{\mathbf{H}}_{ji}\mathbf{V}_{i}\mathbf{P}_{s,i}\mathbf{V}_{i}^{H} \hat{\mathbf{H}}_{ji}^{H}\right)  \\
	& \text{s.t.} ~ tr\left ( \mathbf{P}_{s,i} \right ) =\gamma_{s,i}P_{i} ,
\end{align}
which are equivalent to
\begin{align}  
    \arg\max_{\mathbf{P}_{s,i}} & \gamma_{i}  P_{i}\mathcal{D}\left(\frac{\text{diag}(\mathbf{V}_{i}^{H}\hat{\mathbf{H}}_{ji}^{H}\hat{\mathbf{H}}_{ji}\mathbf{V}_{i})}{tr\left(\mathbf{V}_{i}^{H}\hat{\mathbf{H}}_{ji}^{H}\hat{\mathbf{H}}_{ji}\mathbf{V}_{i}\right)}\right),\label{eq:PsiOpt} \\
	& \text{s.t.} ~ tr\left ( \mathbf{P}_{s,i} \right ) =\gamma_{s,i}P_{i} ,
\end{align}
where $i\in\left\{ A,B\right\} $, $j\in\left\{ A,B\right\} \setminus\left\{ i\right\} $,
and the operator $\text{diag}\left(\mathbf{A}\right)$ provides a
vectors of diagonal entries of $\mathbf{A}$. Values in $\mathbf{P}_{\underline{w},i}$
for the two cases can be obtained as follows. It can be noted that
the constraint, the values in $\mathbf{P}_{s,i},\mathbf{P}_{\underline{w},i},i\in\left\{ A,B\right\} $
should be non-negative, is adopted implicitly.
\paragraph{Case $\theta=1$ (unknown AN)}
In this case, we consider the minimization problem as
\begin{align}  
    \min_{\mathbf{P}_{\underline{w},i},\mathbf{P}_{w,i}} &  tr\left(\hat{\mathbf{H}}_{ji}\left(\mathbf{V}_{i}\mathbf{P}_{w,i}\mathbf{V}_{i}^{H}+\underline{\mathbf{V}}_{i}\mathbf{P}_{\underline{w},i}\underline{\mathbf{V}}_{i}^{H}\right)\hat{\mathbf{H}}_{ji}^{H}\right) , \\
	& \text{s.t.} ~ tr\left(\mathbf{P}_{\underline{w},i}+\mathbf{P}_{w,i}\right) =\left(1-\gamma_{i}\right)P_{i},i\in\left\{ A,B\right\},
\end{align}
Expressing in terms of diagonal values, that is, $\mathbf{p}^{T}=\left[\text{diag}(\mathbf{P}_{\underline{w},i})^{T},\text{diag}(\mathbf{P}_{w,i})^{T}\right]$
and $\mathbf{p}_{0}^{T}=\left[\text{diag}\left(\underline{\mathbf{V}}_{i}^{H}\hat{\mathbf{H}}_{ji}^{H}\hat{\mathbf{H}}_{ji}\underline{\mathbf{V}}_{i}\right)^{T},\text{diag}(\mathbf{V}_{i}^{H}\hat{\mathbf{H}}_{ji}^{H}\hat{\mathbf{H}}_{ji}\mathbf{V}_{i})\right]$,
we have 
\begin{align}
	\label{eq:AN_power_allocation}
    \mathbf{p}=\arg\min_{\mathbf{1}^{T}\mathbf{p}=\left(1-\gamma_{i}\right)P_{i},\mathbf{p}\geq0}~\mathbf{p}^{T}\mathbf{p}_{0} =\left(1-\gamma_{i}\right)P_{i}\mathbf{e}_{k},
\end{align}
where $k=\arg\min_{k}\mathbf{p}_{0}(k)$. This method allocates power
to one of the streams, thus, it is not efficient from the perspective of
minimizing Eve rates. We define $\xi_i = \frac{P_{w,i}}{P_i}$ to describe the ANs' power allocation policy, such that $P_{w,i} = \xi_i P_i$ and $P_{\underline{w},i} = \left ( 1-\xi_i \right )P_i$.

Another way of choosing unequal power allocation for artificial noise
is to choose based on the inverse of eigenvalues of $\hat{\mathbf{H}}_{ji}$. Performing eigenvalue decomposition such that
\begin{align}
	\mathbf{V}_{i}^{H}\hat{\mathbf{H}}_{ji}^{H}\hat{\mathbf{H}}_{ji}\mathbf{V}_{i}=\mathbf{U}_{w,i}\Lambda_{w,i}\mathbf{U}_{w,i}^{H},
\end{align}
the power can be allocated as
\begin{align}
	\label{eq:eig_power_allocation}
	\mathbf{P}_{w,i}=\frac{1-\gamma_{i}}{2}P_{i}\frac{\Lambda_{w,i}^{-1}}{\mathbf{1}^{T}\Lambda_{w,i}^{-1}\mathbf{1}} .
\end{align}
 
\paragraph{Case $\theta=0$ (known AN) }
Since $\mathbf{C}_{ji}$ does not contain the term of $\mathbf{P}_{\underline{w},i}$,
thus, we can choose these diagonal values uniformly, that is
\begin{align}
	\mathbf{P}_{w,i}=\frac{1-\gamma_{i}}{2}P_{i}\frac{1}{b}\mathbf{I},~\text{and}~\mathbf{P}_{\underline{w},i}=\frac{1-\gamma_{i}}{2}P_{i}\frac{1}{n_i}\mathbf{I} .
\end{align}
However, Eve is aware of this allocation, therefore, to improve secrecy
rates, we choose these matrices similar to the Case $\theta=1$.

\section{Simulation Results\label{sec:Simulation-Results}}
A 2D topology model is considered with $(0,0)$ location for Alice
and $(0,1)$ for Bob, while the default location for Eve is set at
$(1,1)$. Path loss exponent is 3 ($\alpha=3$). Rates are averaged
over 100 channel realizations. $P_{A}=P_{B}=25\text{ dB}$, unless
otherwise stated. The channel estimation error variance is assumed
to be $\sigma_{\Delta,BA}^{2}=\sigma_{\Delta,AB}^{2}=0.1$. Noise
variance is set to unity, $\sigma^{2}=1$. The MIMO setting is given as $N_{A}=N_{B}=\frac{N_{E}}{2}=2b=4$. Other parameters depend on the considered scenarios, which are given as follow.
\begin{itemize}
	\item Fixed power allocation, known precoder at Eve, and known AN at Alice and Bob: $\gamma=0.8$, $\kappa=0$, and $\theta=0$.
	\item Known $\beta_{Ei}$, known precoder at Eve, and known AN at Alice and Bob: $\kappa=0$, $\theta=0$, and $\xi=0.9$ for ANs' power allocation.
	\item Known $\beta_{Ei}$, known precoder at Eve, and unknown AN at Alice and Bob: $\kappa=0$, $\theta=1$, and $\xi=0.9$ for ANs' power allocation.
	\item Known $\beta_{Ei}$, partial precoder at Eve, and known AN at Alice and Bob: $\kappa=0.1$, $\theta=0$, and $\xi=0.9$ for ANs' power allocation.
	\item Known $\beta_{Ei}$, partial precoder at Eve, and unknown AN at Alice and Bob: $\kappa=0.1$, $\theta=1$, and $\xi=0.9$ for ANs' power allocation.
	\item Unknown $\beta_{Ei}$, and known precoder at Eve: $\gamma = 1$, and $\kappa = 0$.
	\item Unknown $\beta_{Ei}$, and partial precoder at Eve: $\gamma = 1$, and $\kappa = 0.1$.
\end{itemize}
In cases of unknown $\beta_{Ei}$, Alice and Bob choose to allocate power to maximize their rates,
that is, $\gamma_{i}=1$. For fixed allocation, $\gamma_{i}=\gamma$
and $\kappa_{iE}=\kappa$ are used.

\subsection{IBFD gain on secrecy rate}
To evaluate the potential of IBFD radios to improve the secrecy rate, Figure \ref{fig:sumSR} compares the achievable secrecy rates of HD and IBFD with equal power allocation and partial precoder available at Eve ($\kappa=0.1$). We reflect the reciprocity level by channel uncertainty in simulations and assume perfect SIC applied ($\eta=0$).
It shows that IBFD can generally improve the secrecy rate when the artificial noise is known ($\theta=0$), and the improvement increases with increasing channel uncertainty. For instance, if the interval between the two probing of HD is long, the channel reciprocity is reduced, while IBFD is free of such concerns.
\begin{figure*}
	\centering
	\subfigure[identical channel uncertainty for IBFD and HD]{
		\begin{minipage}[t]{0.5\linewidth}
			\centering
			\includegraphics[width=1.0\linewidth]{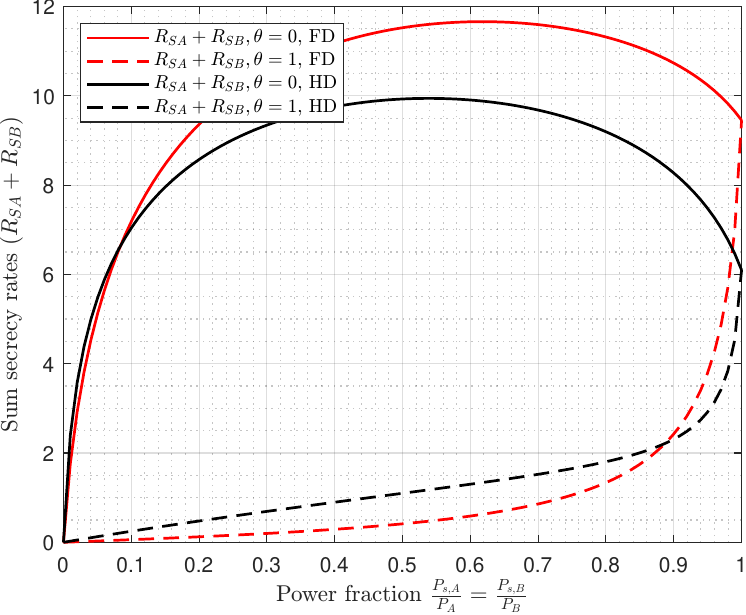}
		\end{minipage}%
	}%
	\subfigure[HD has 1.5 times higher channel uncertainty than IBFD]{
		\begin{minipage}[t]{0.5\linewidth}
			\centering
			\includegraphics[width=1.0\linewidth]{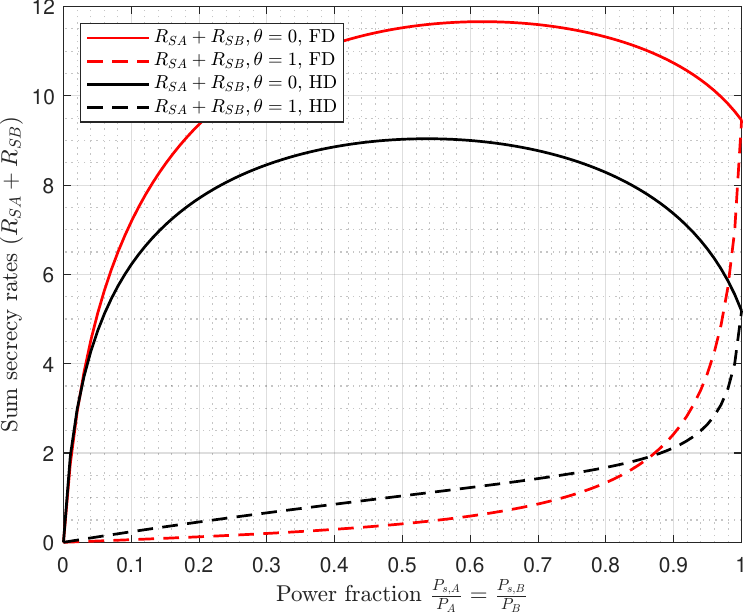} 
		\end{minipage}%
	}%
	\centering
	\caption{Sum secrecy rates with $P_{s,i}=1-P_{w,i}-P_{\underline{w},i}$, $P_{w,i}=P_{\underline{w},i}$.\label{fig:sumSR}}
\end{figure*}

\subsection{Sum secrecy rates versus RSICs $\left(\eta\right)$}
The performance of IBFD radios is subjected to the self-interference cancellation depths. To illustrate its effects on the secrecy rates for different cases, Figure \ref{fig:sum-eta} shows secrecy rates versus the strength of residual SI channels $(\eta)$.
It can be observed that as $\eta$
is increased, secrecy rates are decreased. Secrecy rates with unknown
AN (red and blue dash-dotted lines) are close to zero when $\eta>0.5$
approximately. With $\gamma=1$ constant cases (black lines), secrecy
rates are positive and saturate when $\eta$ is high. On the other
hand, with known AN (red and blue solid lines), optimum power allocation
$(\gamma\neq1)$ provides decreasing rates with increase in $\eta$.
It can also be noted that the plot with $\gamma=0.8$ (green dashed
line) is above other plots with known AN and power allocation. This
is due to the fact that the power allocation is selected based on
ergodic rate approximation, whereas the rates are plotted and averaged
with instantaneous channel estimates. In other words, for $\eta=2$,
$\gamma=0.8$, $\tilde{R}_{Ei}<\tilde{R}_{ji}$; however, $\mathbb{E}\left\{ R_{Si}\right\} >0$.
This shows the limitation of the approximation due to the unavailability
of Eve's CSI at Alice and Bob. 
\begin{figure}
	\centering \includegraphics[width=1.0\columnwidth]{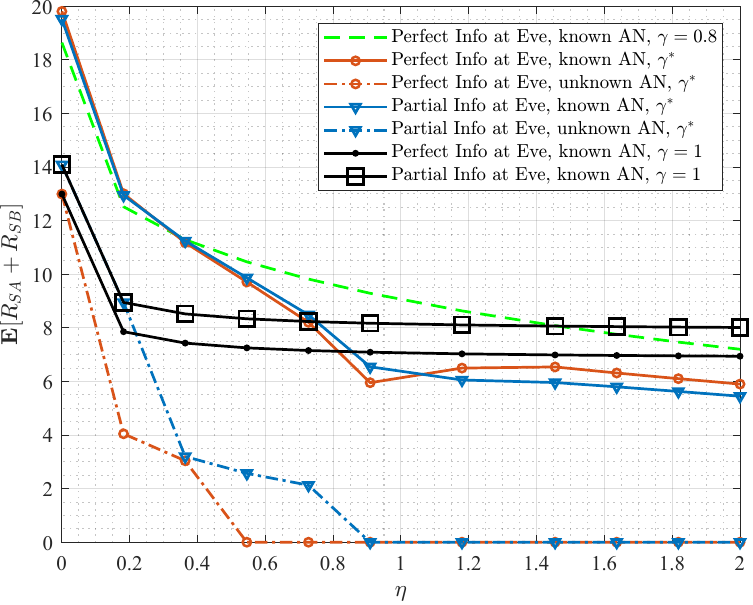}
	\caption{Averaged sum secrecy rates versus strength of residual self-interference channels for different cases. \label{fig:sum-eta}}
\end{figure}

\subsection{Suboptimal fine power allocation}
To evaluate the suboptimal fine power allocation policies, Figure \ref{fig:Eve-rates-fine} plots Eve's rates $\left(R_{EA}\right)$ averaged over 100 realizations for fixed $\gamma=0.5$ and varying fraction for the two artificial noises $P_{w,i}$ and $P_{\underline{w},i}$.
The trend for $R_{EB}$ is similar to $R_{EA}$, so it is omitted here. The fraction $0\leq\xi\leq1$ regulates the ANs' power $(1-\gamma_{i})P_i$
between $P_{w,i}$ and $P_{\underline{w},i}$, that is, $\frac{P_{\underline{w},i}}{P_{w,i}}=\frac{1-\xi}{\xi}$.
It can be observed that allocations for $\mathbf{P}_{s,i}$ given in Equation \eqref{eq:AN_power_allocation}
provides higher eavesdropping rates than the eigenvalue decomposition-based allocation given in Equation \eqref{eq:eig_power_allocation}. In this case, it is desired to allocate more power to the AN injected into the null space instead of the signal space to minimize the eavesdropping rates. With the eigenvalue decomposition-based allocation method, we can adaptively adjust the ANs' power ratio $\xi$ depending on the true conditions.
\begin{figure}
	\centering \includegraphics[width=1.0\columnwidth]{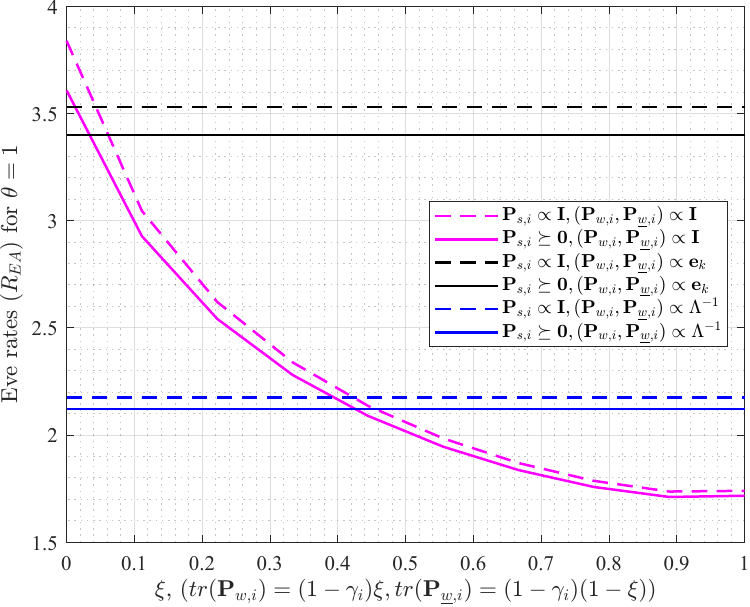}
	\caption{Eve rates $\left(R_{EA}\right)$ for different fine allocation methods with $\theta=1$ (Alice $(0,0)$, Bob $(0,1)$ and Eve $(0.5,5)$).\label{fig:Eve-rates-fine}}
\end{figure}

\subsection{Sum secrecy rates with different Eve's locations }
Figure \ref{fig:sum-x} plots averaged sum secrecy rates for different
scenarios mentioned above while varying $x$-location of Eve. 
\begin{figure}
	\centering \includegraphics[width=1.0\columnwidth]{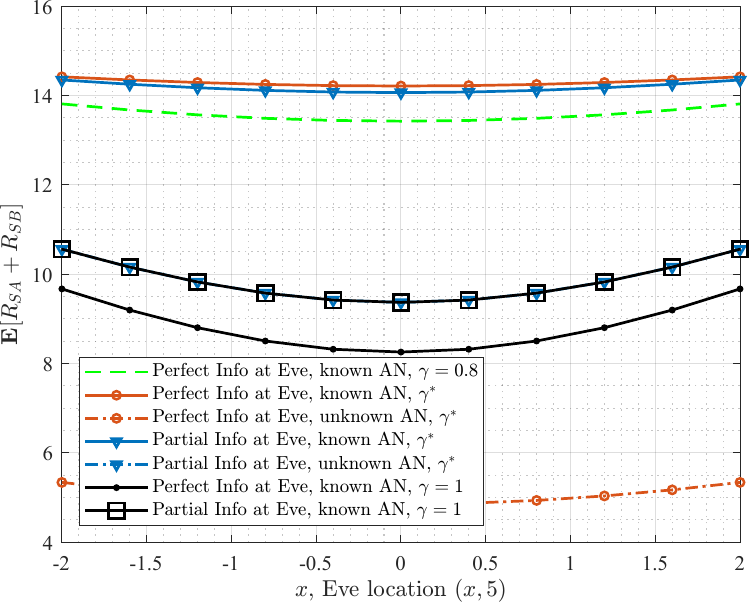}
	\caption{Averaged sum secrecy rates versus different $x$-locations of Eve $\left(x,5\right)$ for different cases. \label{fig:sum-x}}
\end{figure}
The following points can be observed. 
\begin{itemize}
	\item Unknown AN $\theta=1$ (dash-dotted lines) reduces the
	rates at Alice and Bob, thereby reducing secrecy rates, as compared
	to the case of known AN $\theta=0$ (solid lines). 
	\item The curve with a constant $\gamma=0.8$ with known AN lies in the middle
	of these curves. This is plotted for comparison.
	\item For known AN ($\theta=0$), secrecy rates are proportional
	to $\kappa$, that is, $R_{Si}\propto\log\left(\kappa \upsilon_{5}+\upsilon_{6}\right)\gamma_{i}$
	(approximately), where $\upsilon_{k}$ are some constants. It shows that
	maximizing secrecy rates allocates more power to the signal when $\kappa$ is increased,
	that is, $\gamma_{i}^{*}\big|_{\kappa=0=\theta}>\gamma_{i}^{*}\big|_{\kappa>0=\theta}$.
	\item For unknown AN ($\theta=1$), Eve's rates are reduced due to partial information, leading to higher secrecy rates (red and blue dash-dotted
	lines). It is due to the fact that at $\theta=1$, $\gamma=1$ is
	the most likely solution, that is, $\gamma_{i}^{*}\big|_{\kappa=0=1-\theta}=1$.
	\item When $\gamma=1$ is constant for $\theta=0$ (black lines), increasing
	$\kappa$ (i.e., reducing partial information at Eve) reduces Eve's rates, in turn, improves secrecy rates. Since for constant $\gamma$,
	Eve's rates remain unchanged for known/unknown AN ($\theta=0,1$),
	so only $\theta=0$ results are plotted. 
\end{itemize}

\subsection{Sum secrecy rates versus transmit power}
Figure \ref{fig:sum-txP} depicts variations of sum secrecy rates
with respect to transmit power for different cases. 
\begin{figure}
	\centering \includegraphics[width=1.0\columnwidth]{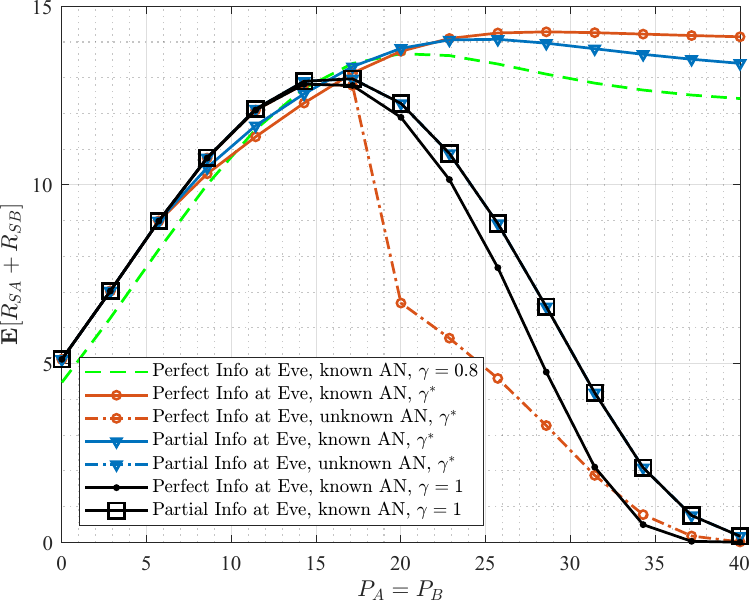}
	\caption{Averaged sum secrecy rates versus transmit SNR for different cases.\label{fig:sum-txP}}
\end{figure}
It can be seen
that secrecy rates achieve the maximum at first but go down to zero in the high-SNR region for some
cases. This is due to the fact that at high SNRs,
the power allocation algorithm provides $\gamma=1$, which suggests stopping transmitting the AN, yielding high eavesdropping rates, as described in Section \ref{subsec:High-SNR-regime}.
The cases that do not yield $\gamma=1$ include cases with ``known
AN'' (solid red and blue lines) and ``fixed power allocation''
(dashed-green line). The rates provided by these cases saturate at
high SNRs, which is due to the reason that self-interference and channel
estimation errors limit the secrecy performance, which is also illustrated in \cite{luo2023channel}. The limited information and non-zero artificial noise power are limiting
factors for the eavesdropping rates, see equations (\ref{eq:highSNR_RBA}) and (\ref{eq:highSNR_REA}).

\section{Conclusion\label{sec:Conclusion} }
In this paper, we have designed an AN-based signal vector,
where the AN is injected in both the signal space and the null space. Thanks to the signal space AN, our design has the potential to further improve the secrecy rate in certain specialized scenarios. These scenarios are typically specialized and outside the scope of conventional communication systems. However, our power allocation scheme can adaptively reduce the power budget of signal space AN to zero for conventional communication systems based on the conditions, which becomes consistent with the typical design. Considering large-scale fading information of the eavesdropping channels available at Alice and
Bob, which could become feasible with the development of ISAC technology, a two-step power allocation has been proposed, which is based on
ergodic rate approximation. Secrecy region and power allocation have
been analyzed in known and unknown AN cases, and corresponding simulation
results have been demonstrated. The numerical results illustrate that unknown AN at legitimate users reduces
the secrecy rates, while the limited information at Eve could improve the secrecy rates.
Due to the lack of availability of Eve's CSI at Alice and Bob, the resulting
power allocation has achieved ergodic secrecy rates close to their approximations.

\section*{Acknowledgment}
This work was supported in part by Huawei Technologies Canada Company Ltd. The work of Tharmalingam Ratnarajah was supported by the U.K. Engineering and Physical Sciences Research Council (EPSRC) under Grant EP/T021063/1.

\appendices

\section{Proof of Lemma \ref{lem:MMSE-rate}\label{subsec:Proof-of-Lemma1}}
\begin{proof}
	Proof can be given as 
	\begin{align*}
	     \mathbb{E}\left \{ R \right \} & =\mathbb{E}\left \{ \log_{2}\left|\mathbf{I}+\mathbf{H}\mathbf{C}_{x}\mathbf{H}^{H}\mathbf{C}_{z}^{-1}\right|   \right \} \\
	     & \stackrel{(a)}{\leq}\log_{2}\left|\mathbf{I}+\mathbf{C}_{x}\mathbb{E}\left\{ \mathbf{H}^{H}\mathbb{E}\left\{ \mathbf{C}_{z}^{-1}\right\} \mathbf{H}\right\} \right|\\
		& \stackrel{(b)}{\approx}\log_{2}\left|\mathbf{I}+\mathbf{C}_{x}\mathbb{E}\left\{ \mathbf{H}^{H}\mathbb{E}\left\{ \mathbf{C}_{z}\right\} ^{-1}\mathbf{H}\right\} \right|,
	\end{align*}
	where $(a)$ is obtained using Jensen's inequality; and in $(b)$,
	$\mathbb{E}\left\{ \mathbf{C}_{z}^{-1}\right\} \approx\mathbb{E}\left\{ \mathbf{C}_{z}\right\} ^{-1}$
	is used. 
\end{proof}

\section{Proof of Corollary \ref{cor:ergodicrate}\label{subsec:Proof-of-Corollary2}}
\begin{proof}
	To avoid the redundancy, the proof is only presented for Eve below
	as (the steps for Alice and Bob can be followed similarly) %
	\begin{align}
		\mathbb{E}\left \{ R_{E} \right \} & \stackrel{(a)}{\approx}\log_{2}\left|\mathbf{I}+\mathbb{E}\left\{ \mathbf{P}_{s}\hat{\mathbf{V}}_{E}^{H}\mathbf{H}_{E}^{H}\mathbb{E}\left\{ \mathbf{C}_{E}\right\} ^{-1}\mathbf{H}_{E}\hat{\mathbf{V}}_{E}\right\} \right| \\
		& \stackrel{(b)}{=}\log_{2}\left|\mathbf{I}+\frac{\mathbb{E}\left\{ \mathbf{P}_{s}\hat{\mathbf{V}}_{E}^{H}\mathbf{H}_{E}^{H}\mathbf{H}_{E}\hat{\mathbf{V}}_{E}\right\} }{c_{E}}\right|\\
		& \stackrel{(c)}{=}\log_{2}\left|\mathbf{I}+\frac{N_{E}\mathcal{BD}\left(\mathbf{P}_{s,A}\beta_{EA},\mathbf{P}_{s,B}\beta_{EB}\right)}{c_{E}}\right|\\
		& \stackrel{(d)}{\approx}\log_{2}\left|\mathbf{I}+\frac{N_{E}\mathcal{BD}\left(\frac{P_{s,A}}{d}\mathbf{I}\beta_{EA},\frac{P_{s,B}}{d}\mathbf{I}\beta_{EB}\right)}{c_{E}}\right|,
	\end{align}
	where $(a)$ is obtained via (\ref{eq:lem-eqn}); in $(b)$ and $(c)$,
	$\mathbb{E}_{H}\mathbf{C}_{E}=c_{E}\mathbf{I}$ and the numerator
	is simplified; in $(d)$, the value $\mathbf{P}_{s,i}=\frac{P_{s,i}}{d}\mathbf{I}$
	is substituted; and simplifying the determinant provides the expression
	in (\ref{eq:RE-appr}) as $\left|\mathcal{BD}\left(\alpha_{1}\mathbf{I}_{d},\alpha_{2}\mathbf{I}_{d}\right)\right|=\alpha_{1}^{d}\alpha_{2}^{d}$. 
\end{proof}

\section{Quadratic constraint equations \label{subsec:Quadratic-cons-equation}}
In this appendix, we give the detailed derivation of the quadratic constraint equations as follows, where $\mathbf{e}_{i}$
denotes the $i^{th}$ column of the identity matrix.
\begin{align}
	    c_{BA} & =\left(P_{s,B}+P_{w,B}+P_{\underline{w},B}\right)\eta+\theta\beta_{BA}\left(P_{w,A}+P_{\underline{w},A}\right)\nonumber \\
	    &\quad +\sigma_{\Delta,BA}^{2}\left(P_{s,A}+P_{w,A}+P_{\underline{w},A}\right)+\sigma^{2} \nonumber \\
		& =\eta P_{B}+\theta\beta_{BA}\left(1-\gamma_{A}\right)P_{A}+\sigma_{\Delta,BA}^{2}P_{A}+\sigma^{2} \nonumber \\
		& =\left[\begin{array}{c}
			\underbrace{-\theta\beta_{BA}P_{A}}_{a_{BA}} \\
			0
		\end{array}\right]^{T}\left[\begin{array}{c}
			\gamma_{A}\\
			\gamma_{B}
		\end{array}\right] \nonumber \\
		&\quad +\underbrace{P_{B}\eta+\theta\beta_{BA}P_{A}+\sigma_{\Delta,BA}^{2}P_{A}+\sigma^{2}}_{g_{BA}} \nonumber \\
		& =a_{BA}\gamma_{A}+g_{BA} = a_{BA}\mathbf{e}_{1}^{T}\mathbf{u}+g_{BA} .
\end{align}
\begin{align}
		c_{AB} & =\left(P_{s,A}+P_{w,A}+P_{\underline{w},A}\right)\eta+\theta\beta_{AB}\left(P_{w,B}+P_{\underline{w},B}\right)\nonumber \\
		&\quad +\sigma_{\Delta,AB}^{2}\left(P_{s,B}+P_{w,B}+P_{\underline{w},B}\right)+\sigma^{2} \nonumber \\
		& =P_{A}\eta+\theta\beta_{AB}\left(1-\gamma_{B}\right)P_{B}+\sigma_{\Delta,AB}^{2}P_{B}+\sigma^{2} \nonumber \\
		& =\left[\begin{array}{c}
			0\\
			\underbrace{-\theta\beta_{AB}P_{B}}_{a_{AB}}
		\end{array}\right]^{T}\left[\begin{array}{c}
			\gamma_{A}\\
			\gamma_{B}
		\end{array}\right]\nonumber \\
		&\quad +\underbrace{P_{A}\eta+\theta\beta_{AB}P_{B}+\sigma_{\Delta,AB}^{2}P_{B}+\sigma^{2}}_{g_{AB}} \nonumber \\
		& =a_{AB}\gamma_{B}+g_{AB}=a_{AB}\mathbf{e}_{2}^{T}\mathbf{u}+g_{AB}.
\end{align}
\begin{align}
		c_{E} & = \beta_{EA}\left(P_{s,A}\kappa_{A,E}+P_{w,A}+P_{\underline{w},A}\right)\nonumber \\
		&\quad+\beta_{EB}\left(P_{s,B}\kappa_{B,E}+P_{w,B}+P_{\underline{w},B}\right)+\sigma^{2} \nonumber \\
		& =\beta_{EA}\left(\gamma_{A}\kappa_{A,E}+1-\gamma_{A}\right)P_{A}\nonumber \\
		&\quad +\beta_{EB}\left(\gamma_{B}\kappa_{B,E}+1-\gamma_{B}\right)P_{B}+\sigma^{2} \nonumber \\
		& =\left[\begin{array}{c}
			\beta_{EA}\left(\kappa_{A,E}-1\right)P_{A}\\
			\beta_{EB}\left(\kappa_{B,E}-1\right)P_{B}
		\end{array}\right]^{T}\left[\begin{array}{c}
			\gamma_{A}\\
			\gamma_{B}
		\end{array}\right]\nonumber \\
		&\quad+\underbrace{\beta_{EA}P_{A}+\beta_{EB}P_{B}+\sigma^{2}}_{g_{E}}=\mathbf{a}_{E}^{T}\mathbf{u}+g_{E}.
\end{align}
\begin{align}
		c_{E}+ N_{E}\beta_{EA}\frac{P_{s,A}}{b}& =\left[\begin{array}{c}
			\beta_{EA}\left(\kappa_{A,E}-1\right)P_{A}+N_{E}\beta_{EA}\frac{P_{A}}{d}\\
			\beta_{EB}\left(\kappa_{B,E}-1\right)P_{B}
		\end{array}\right]^{T}\nonumber \\
		&\quad \left[\begin{array}{c}
			\gamma_{A}\\
			\gamma_{B}
		\end{array}\right]+\beta_{EA}P_{A}+\beta_{EB}P_{B}+\sigma^{2} \nonumber \\
		& =\mathbf{a}_{E}^{T}\mathbf{u}+\underbrace{N_{E}\beta_{EA}\frac{P_{A}}{b}}_{c_{EAB}}\gamma_{A}+ g_E \nonumber \\
		& = \mathbf{a}_{E}^{T}\mathbf{u}+ c_{EAB} \mathbf{e}_{1}^{T}\mathbf{u}+ g_E
		.
\end{align}

\ifCLASSOPTIONcaptionsoff
  \newpage
\fi

\bibliographystyle{ieeetr} 

\end{document}